\documentclass[british]{article}

\usepackage{authblk}

\usepackage{libertineRoman}
\usepackage{biolinum} 
\usepackage[T1]{fontenc}
\usepackage[latin9]{inputenc}
\usepackage{geometry}
\usepackage{color,graphicx}
\usepackage{array}
\usepackage{prettyref}
\usepackage{refstyle} 
\usepackage{booktabs}
\usepackage{amsmath}
\usepackage{amsthm}
\usepackage{amssymb}

\definecolor{darkblue}{RGB}{0,0,158}
\definecolor{lightgrey}{RGB}{220,220,220}

\makeatletter
\newcommand*{\rom}[1]{\mathrm{\expandafter\@slowromancap\romannumeral #1@}}
\usepackage{tabularx}
\newcommand{\PB}{p_B^*} 
\newcommand{\PA}{p_A^*} 
\newcommand{\eps}{\epsilon}

\long\def\old#1{}

\newcommand{\Pos}{\mathsf{Pos}}

\newcommand{\tr}{\mathrm{Tr}} 

\newcommand{\Pst}{{\cal{P}}^{\rm st}}
\newcommand{\epsst}{\epsilon^{\rm st}}
\newcommand{\Peps}{{\cal{P}}^{\epsilon}}
\newcommand{\calP}{{\cal{P}}}
\newcommand{\calZ}{{\cal{Z}}}

\usepackage{algorithmic}

\usepackage{algorithm} 

\newenvironment{varalgorithm}[1]
  {\algorithm[h] 
  	\renewcommand{\thealgorithm}{\mbox{#1}} 
	}
  {\endalgorithm}

\floatname{algorithm}{Protocol}

\usepackage{IEEEtrantools}

\usepackage{braket} 
 
\AtBeginDocument{\providecommand\Eqref[1]{\ref{Eq:#1}}}
\AtBeginDocument{\providecommand\Defref[1]{\ref{Def:#1}}}
\AtBeginDocument{\providecommand\Claimref[1]{\ref{Claim:#1}}}
\AtBeginDocument{\providecommand\Algref[1]{\ref{Alg:#1}}}
\AtBeginDocument{\providecommand\Subsecref[1]{\ref{Subsec:#1}}}
\AtBeginDocument{\providecommand\Secref[1]{\ref{Sec:#1}}}
\AtBeginDocument{}
\AtBeginDocument{\providecommand\Lemref[1]{\ref{Lem:#1}}}
\AtBeginDocument{\providecommand\Figref[1]{\ref{Fig:#1}}}
\AtBeginDocument{\providecommand\Thmref[1]{\ref{Thm:#1}}}
\AtBeginDocument{}
\AtBeginDocument{}

\RS@ifundefined{subsecref}
  {\newref{subsec}{name = \RSsectxt}}
  {}
\RS@ifundefined{thmref}
  {\def\RSthmtxt{theorem~}\newref{thm}{name = \RSthmtxt}}
  {}
\RS@ifundefined{lemref}
  {\def\RSlemtxt{lemma~}\newref{lem}{name = \RSlemtxt}}
  {}

\theoremstyle{plain}
\newtheorem{thm}{\protect\theoremname}
\theoremstyle{definition}
\newtheorem{defn}[thm]{\protect\definitionname}
\theoremstyle{plain}
\newtheorem{assumption}[thm]{\protect\assumptionname}
\theoremstyle{remark}
\newtheorem{claim}[thm]{\protect\claimname}
\theoremstyle{plain}
\newtheorem{lyxalgorithm}[thm]{\protect\algorithmname}
\theoremstyle{plain}
\newtheorem{lem}[thm]{\protect\lemmaname}
\theoremstyle{plain}
\newtheorem{conjecture}[thm]{\protect\conjecturename}
\theoremstyle{definition}

\theoremstyle{remark}

\theoremstyle{plain}

\theoremstyle{plain}
\newtheorem{prop}[thm]{\protect\propositionname}
\theoremstyle{plain}
\newtheorem{cor}[thm]{\protect\corollaryname}

\usepackage{color} 
\newcommand\branchcolor[2]{{\color{#1} #2}}

\usepackage{hyperref}

\hypersetup{
    colorlinks=true,
    linkcolor=blue,
    filecolor=magenta,      
    urlcolor=cyan,
	citecolor=blue}

\hypersetup{hidelinks,
backref=true,
pagebackref=true,
hyperindex=true,
breaklinks=true,
colorlinks=true,%
urlcolor=blue,
bookmarks=true,
bookmarksopen=false}

\newref{thm}{name=theorem~,Name=Theorem~,names=theorems~,Names=Theorems~}
\newref{def}{name=definition~,Name=Definition~,names=definitions~,Names=Definitions~}
\newref{alg}{name=protocol~,Name=Protocol~,names=protocols~,Names=Protocols~}
\newref{cor}{name=corollary~,Name=Corollary~,names=corollaries~,Names=Corollaries~}
\newref{lem}{name=lemma~,Name=Lemma~,names=lemmas~,Names=Lemmas~}
\newref{claim}{name=claim~,Name=Claim~,names=claims~,Names=Claims~}
\newref{sec}{name=section~,Name=Section~,names=sections~,Names=Sections~}
\newref{subsec}{name=section~,Name=Section~,names=sections~,Names=Sections~}
\newref{subsubsec}{name=section~,Name=Section~,names=sections~,Names=Sections~}
\newref{prop}{name=proposition~,Name=Proposition~,names=propositions~,Names=Propositions~}
\newref{conj}{name=conjecture~,Name=Conjecture~,names=conjectures~,Names=Conjectures~}
\newref{assu}{name=assumption~,Name=Assumption~,names=assumptions~,Names=Assumptions~}
\newref{rem}{name=remark~,Name=Remark~,names=remarks~,Names=Remarks~}
\newref{fact}{name=fact~,Name=Fact~,names=facts~,Names=Facts~}
\newref{exa}{name=example~,Name=Example~,names=examples~,Names=Examples~}

\@ifundefined{showcaptionsetup}{}{
 \PassOptionsToPackage{caption=false}{subfig}}
\usepackage{subfig}

\global\long\def\tr{\text{tr}}

\makeatother

\usepackage{babel}
\providecommand{\algorithmname}{Protocol}
\providecommand{\assumptionname}{Assumption}
\providecommand{\claimname}{Claim}
\providecommand{\conjecturename}{Conjecture}
\providecommand{\definitionname}{Definition}
\providecommand{\examplename}{Example}
\providecommand{\lemmaname}{Lemma}
\providecommand{\remarkname}{Remark}
\providecommand{\theoremname}{Theorem}
\providecommand{\factname}{Fact}
\providecommand{\propositionname}{Proposition}
\providecommand{\corollaryname}{Corollary}

\begin{document}
\title{Improving device-independent weak coin flipping protocols}
\author[1,2,3]{Atul Singh Arora}
\author[4,5]{Jamie Sikora}
\author[6,7,8]{Thomas Van Himbeeck}
\affil[1]{Joint Center for Quantum Information and Computer Science (QuICS), University of Maryland \& NIST, College Park, Maryland, USA}
\affil[2]{Department of Computing and Mathematical Sciences, California Institute of Technology, USA}
\affil[3]{Institute for Quantum Information and Matter, California Institute of Technology, USA}
\affil[4]{Virginia Polytechnic Institute and State University, USA} 
\affil[5]{Perimeter Institute for Theoretical Physics, Canada} 
\affil[6]{University of Toronto, Canada}
\affil[7]{Institute of Quantum Computing, University of Waterloo, Canada}
\affil[8]{T\'{e}l\'{e}com Paris, France} 
\date{April 2024}                    
\setcounter{Maxaffil}{0}
\renewcommand\Affilfont{\itshape\small}

\pagenumbering{roman} 

\maketitle
\begin{abstract}
Weak coin flipping is the cryptographic task where Alice and Bob remotely flip a coin but want opposite outcomes. 
This work studies this task in the device-independent regime where Alice and Bob neither trust each other, nor their quantum devices. 
The best protocol was devised over a decade ago by Silman, Chailloux, Aharon, Kerenidis, Pironio, and Massar with bias $\varepsilon \le 0.33664$, where the bias is a commonly adopted security measure for coin flipping protocols. This work presents two techniques to lower the bias of such protocols, namely self-testing and abort-phobic compositions. We apply these techniques to the SCAKPM '11 protocol above and, assuming a continuity conjecture, lower the bias to $\varepsilon \approx  0.29104$. We believe that these techniques could be useful in the design of device-independent protocols for a variety of other tasks. 

Independently of weak coin flipping, en route to our results, we show how one can test $n-1$ out of $n$ devices, and estimate the performance of the remaining device, for later use in the protocol. The proof uses linear programming and, due to its generality, may find applications elsewhere.

\end{abstract}

\pagebreak 
\tableofcontents

\pagebreak 

\pagenumbering{arabic} 

\section{Introduction}  
Coin flipping is the two-party cryptographic primitive where two parties, henceforth called Alice and Bob, wish to flip a coin, but %
{where---and this is the non-trivial requirement---they do not trust each other.}
This primitive was introduced by Blum~\cite{Blum} who also introduced the first (classical) protocol { achieving this functionality}. 
In this work, we concentrate on \emph{weak} coin flipping (WCF) protocols where Alice and Bob desire opposite outcomes. Since then, a series of quantum protocols were introduced with successively improved security. Mochon, in his tour de force, finally settled the question about the limits of the security in the quantum regime by proving the \emph{existence} of quantum protocols with security approaching the ideal limit~\cite{Mochon07}. 
This was followed by a flurry of results which achieved diverse cryptographic functionality assuming WCF as a black-box, such as strong coin flipping~\cite{CK09}, bit commitment~\cite{CK11}, a variant of oblivious transfer~\cite{Chailloux2013a}, leader election~\cite{Ganz2009} and dice rolling~\cite{Aharon2009}, establishing the importance of WCF in the quantum setting. 
Returning to Mochon, his work was quite technical and based on the notion of \emph{point games}, a concept introduced by Kitaev. Interestingly, his work was never published---only a preprint was available. 
Subsequently, a sequence of works have continued the study of point games. 
In particular, the proof of existence was eventually simplified and peer reviewed~\cite{ACG+14} and explicit protocols were reported after more than a decade of Mochon's work~\cite{Arora2019,ARV21}.\footnote{\cite{arora2024protocols} is the distilled concise version of these works and \cite{cryptoeprint:2022/1101} is the comprehensive version subsuming both works. Interestingly, Miller~\cite{Miller2019} used techniques from Mochon's proof to show that protocols approaching the ideal limit must have an exponentially increasing number of messages.} %
Yet, we note that all of this work is in the \emph{device-dependent} setting where \emph{Alice and Bob trust their quantum devices}. Very little is known in the \emph{device-independent} setting where a cheating player is allowed to control an honest player's quantum devices, opening up a plethora of new cheating strategies that were not considered in the previously mentioned references.  

We introduce some basic concepts to facilitate further discussion. The prefix \emph{weak} in weak coin flipping refers to the situation where Alice and Bob desire opposite outcomes of the coin. (We have occasion to discuss \emph{strong} coin flipping protocols, where Alice or Bob could try to bias the coin towards either outcome, but it is not the focus of this work.) 
When designing weak coin flipping protocols, the security goals are as follows. 
\bigskip 
\begin{center} 
\begin{tabularx}{\textwidth}{cX} 
\emph{Correctness for honest parties:} & If Alice and Bob are honest, then they share the same outcome of a protocol $c \in \{ 0, 1 \}$, and $c$ is generated uniformly at random by the protocol. \\

\emph{Soundness against cheating Alice:} & If Bob is honest, then a dishonest (i.e., cheating) Alice cannot force the outcome $c = 0$. \\
\emph{Soundness against cheating Bob:} & If Alice is honest, then a dishonest (i.e., cheating) Bob cannot force the outcome $c = 1$. 
\end{tabularx} 
\end{center}  
\bigskip 

The commonly adopted goal of two-party protocol design is to assume perfect correctness and then minimize the effects of a cheating party, i.e., to make it as sound as possible. 
This way, if no parties cheats, then the protocol at least does what it is meant to still. 
With this in mind, we need a means to quantify the effects of a cheating party. 
It is often convenient to have a single measure to determine if one protocol is better than another. 
For this purpose, we use \emph{cheating probabilities} (denoted $\PB$ and $\PA$) and \emph{bias} (denoted $\eps$), defined as follows. 

\bigskip 
\begin{center} 
\begin{tabularx}{\textwidth}{cX}
$\PA$: & The maximum probability with which a dishonest Alice can force an honest Bob to accept the outcome $c = 0$. \\ 
$\PB$: & The maximum probability with which a dishonest Bob can force an honest Alice to accept the outcome $c = 1$. \\ 
$\varepsilon$: & The maximum amount with which a dishonest party can bias the probability of the outcome away from uniform. Explicitly, $\varepsilon = \max \{ \PA, \PB \} - 1/2$. \\ 
\end{tabularx} 
\end{center}  
\bigskip 

These definitions are not complete in the sense that we have not yet specified what a cheating Alice or a cheating Bob are allowed to do, or of their capabilities.
In this work, we study \emph{information theoretic security}---Alice and Bob are only bounded by the laws of quantum mechanics. 
For example, they are not bounded by polynomial-time quantum computations. 
In addition to this, we study the security in the \emph{device-independent} regime where we assume Alice and Bob have complete control over the quantum devices when they decide to ``cheat''.

When studying device-independent (DI) protocols, one should first consider whether or not secure classical protocols are known (since these are not affected by the DI assumption). 
It was proved that every classical WCF protocol\footnote{also holds for strong coin flipping} has bias $\varepsilon = 1/2$, which is the worst possible value (see~\cite{Kitaev03,HW11}). 
Thus, it makes sense to study quantum WCF protocols in the DI setting, especially if one with bias $\varepsilon < 1/2$ can be found. Indeed, Silman, Chailloux, Aharon, Kerenidis, Pironio, and Massar presented a protocol (see \Algref{SCF}) in~\cite{Silman2011} with $p^*_A=\cos^2(\pi/8)\approx 0.853$ and $p^*_B=3/4$. We briefly discuss this protocol because we build on this result but defer the details. To this end, denote by \emph{boxes}, the spatially separable constituents of an untrusted \emph{quantum device}, each of which accepts a classical input and produces a classical output. For instance, an untrusted quantum device corresponding to the GHZ game\footnote{A GHZ game is a 3-player non-local game where each player is asked a single bit question and produces a single bit single bit answer; we review this in \Secref{FirstTechSelfTest}.} consists of three \emph{boxes}, each accepting and outputting a single bit. Returning to the protocol in ~\cite{Silman2011}, \Algref{SCF} starts with Alice possessing \emph{two} boxes and Bob possessing \emph{one} box which are together supposed to contain the GHZ state and measurements.\footnote{They specify the best quantum strategy for winning the GHZ game.} As the protocol proceeds, they, in addition to exchanging classical information, operate these boxes and exchange them.\footnote{Any protocol described using boxes is readily converted into one where Alice and Bob communicate over an insecure quantum channel; see \Secref{BoxParadigm}.} 
As is, \Algref{SCF} has bias $\varepsilon \approx 0.353$ but in \cite{Silman2011}, \Algref{SCF} is composed many times to 
lower the bias to $\varepsilon \le 0.33664$.\footnote{While here compositions are done in a specific context, universal composition of weak coin flipping protocols has recently been studied in \cite{wu2024composable}.}

\subsection{Main Result}

In this work, we provide two techniques for lowering the bias of weak coin flipping protocols and apply them to (the natural weak coin flipping variant of) \Algref{SCF}, to obtain the following. %

	\begin{thm} 
	There exist device-independent weak coin flipping protocols with bias, $\varepsilon$, approaching $0.29014$, assuming two continuity conjectures, \Conjref{Qcont,Pcont}, hold. %
	\end{thm} 

	Before discussing the proof, we note that \Algref{SCF} was, in fact, a strong coin flipping protocol and we begin by turning it into a weak coin flipping protocol---\Algref{WCF}---in the most natural way. %
	Again, we defer the explicit description of the protocol and informally describe the basic idea: since weak coin flipping has the notion of a ``winner'' (if $c=0$ Alice wins and if $c=1$ Bob wins) we have the party who does not win, conduct an additional test. 
	
	The proof of our theorem relies on two key techniques. Our first technique is to add a pre-processing step to \Algref{WCF} which \emph{self-tests} the boxes shared by Alice and Bob at the start of the protocol. 
Our second technique is to compose and analyse the resulting protocols in a new way,\footnote{The composition in \cite{Silman2011} may also be seen as "abort-phobic" but their analysis doesn't rely on the "abort" probability; their bound essentially neglects the abort event.} which we call \emph{abort-phobic} composition.

	\subsection{First technique: Self-testing} 

	In the original Protocol~\ref{alg:SCF} and its WCF variant, Protocol~\ref{alg:WCF}, a cheating party may control what measurement is performed in the boxes of the other party and how the state of the boxes is correlated to its own quantum memory. This is more general than \textit{device-dependent} protocols, where for instance, the measurements are known to the honest player. 
	However, we employ the concept of self-testing to stop Bob (or Alice) from applying such a strategy. Intuitively, self-testing is a powerful property which allows one to, just from certain input-output behaviours of given devices (satisfying minimal assumptions), conclude uniquely which quantum states and measurements constitute the devices (up to relabelling). The GHZ state which was used in \Algref{SCF,WCF} can be self-tested. Clearly, this property has the potential to improve their security.\footnote{In \cite{Silman2011}, it was noted that self-testing doesn't help improve the security of \Algref{SCF}. Alternatively stated, \Algref{SCF} has the curious property that its device dependent variant has the same security as it (the device dependent variant).}
	
	We define two variants of \Algref{WCF}: \Algref{AliceSelfTests}, where Alice self-tests Bob before executing \Algref{WCF}, and \Algref{BobSelfTests}, where Bob self-tests Alice instead. Skipping the details, the basic construction is almost trivial. Alice and Bob start with $n$ triples of boxes (constituting $n$ untrusted quantum devices). When Alice self-tests, for instance, Alice asks Bob to send all but one randomly selected triple and tests if the GHZ test passes for these. If so, the remaining triple is used for the actual protocol. 
	A large enough $n$ forces a dishonest Bob to not tamper with the boxes too much, as suggested above.
	Indeed, \Algref{AliceSelfTests} (i.e. when Alice self-tests) already allows us to reduce the cheating probabilities.\footnote{\Algref{BobSelfTests} (i.e. when Bob self-tests) does not result in a lower bias. However, as we show, when protocols are composed using our second technique, \Algref{BobSelfTests} helps lower the bias.}

	\begin{prop} [Informal. See~\Propref{AliceSelfTests} for a formal statement] 
	For \Algref{AliceSelfTests}, i.e. where Alice self-tests Bob, the cheating probabilities, in the limit of $n\to \infty$, are 
	\begin{equation} 
	\PA = \cos^2 (\pi/8) \approx 0.85355 \quad \text{ and } \quad \PB \approx 0.6667,
	\end{equation} 
	assuming a continuity conjecture.
	\end{prop}  
	For comparison, recall that for \Algref{SCF} (it turns out, also for \Algref{WCF}), $\PA=\cos^2(\pi/8)$ and $\PB=3/4$. We prove this lemma in two stages. In the \emph{first} stage (see \Secref{securityAsymptotic}), we assume perfect self-testing: the self-testing step results in exactly specifying (up to a relabelling) the state and measurements governing Alice's boxes. This may be seen as taking $n \to \infty$ in the self-testing step. It is known that for device-dependent protocols, where Alice and Bob trust their devices, the cheating probabilities can be cast as values of semi-definite programs (SDPs) ~\cite{Kitaev03,Mochon07}. Perfect self-testing allows us to, therefore, express Bob's cheating probabilities as an SDP. Its numerical evaluation yields the quoted value. Analysis for Alice's cheating probability is unchanged from \Algref{WCF}. 
	{In the \emph{second} stage (see \Secref{SecurityFiniteN}), we analyse the protocol with $n$ finite and prove that it converges to the case above, under a precisely stated continuity conjecture. The analysis consists of the following two key conceptual steps.}

	{
	\paragraph{Self-testing in a cryptographic setting.} Since Alice tests $n-1$ devices, at best she can conclude (with some confidence) that the last device wins the GHZ game with probability $1-\eps(n)$ where $\eps(n)$ decreases with $n$. We therefore need a \emph{robust} self-testing result that allows one to conclude, in particular, the following: if the success probability of a device in a GHZ test is close to unity, then the states and measurements constituting the device are close to GHZ states and measurements (up to a relabelling), in say trace distance. Fortunately, such \emph{robust} self-testing results are known for the GHZ game~\cite{MillerShi} and so it only remains to show the first statement, i.e. estimating the success probability of the remaining device, based on testing $n-1$ devices. While this looks straightforward, formalising and proving this statement turns out to be a bit subtle (see \Subsecref{EstimateGHZ}). For instance, related statements are known (e.g.~\cite{Vazirani_2014}) when all devices are measured but they do not apply to our setting where one device is left for later use in the protocol. Our statement holds quite generally for any game with perfect completeness,\footnote{i.e. there is a quantum strategy for winning the game with probability 1} and is proved using linear programming. It may, therefore, be of independent interest. In particular, our statement holds for both \Algref{AliceSelfTests,BobSelfTests}.}
	
	{\paragraph{Continuity conjecture.} Consider \Algref{AliceSelfTests} where Alice self-tests. Using the result described above, one can conclude that the state and measurements in the remaining pair of boxes held by Alice, are $\eps$-close (in trace distance) to the GHZ state and measurements (up to relabelling). If they were exactly the same, the cheating probabilities could be expressed as the value, say $\eta_0$, of an SDP. When they are not, one can still write an optimisation problem whose solution, say $\eta(\eps)$, bounds the cheating probabilities, but it is unclear if this optimisation problem is an SDP. The main issues are that the adversary can use states of arbitrary dimensions (that could scale with $\eps$) and that the optimisation is over both states and measurements. We conjecture that $\lim_{\eps\to 0} \eta(\eps) = \eta_0$, i.e. the analysis with $\eps$-close boxes converges to the analysis done assuming Alice holds GHZ boxes. Since $\eta(0)=\eta_0$, if the conjecture fails, it must mean there is a discontinuity in $\eta(\eps)$ near $\eps=0$ and this would be very surprising. The precise continuity statements for \Algref{AliceSelfTests,BobSelfTests} are stated as \Conjref{Pcont,Qcont}, respectively. }
	
	\vspace{1em}

	{While we leave the proofs of these conjectures to future work, we remark that we do not see any obvious obstacles for proving \Conjref{Pcont}. \Conjref{Qcont} (when Bob self-tests), on the other hand, seems more involved as the box held by Bob is not measured right after the self-test step, but only later---after some interaction has taken place between the two parties. This makes it more difficult to relax the optimisation problem (corresponding to the cheating probabilities) to obtain an SDP in dimensions independent of $\eps$. }

	\subsection{Second technique: abort-phobic composition}\label{subsec:IntroSecondTechnique}
    It can happen, that for a given WCF protocol, $\PB \neq \PA$, in which case we say the protocol is \emph{polarised}.  
    As we saw earlier, it is known (e.g. \cite{Silman2011}) that composing a polarised protocol with itself (or other protocols) can effectively reduce the bias. Our second improvement is a modified way of composing protocols when there is a {non-zero} %
	probability that the honest player catches the cheating player. 
    Let us start by recalling the standard way of composing protocols.

	{\paragraph{Standard composition.} 
	For a protocol with cheating probabilities $\PB$ and $\PA$, we say that it has polarity towards Alice when it satisfies $\PA > \PB$. 
	Similarly, we say that it has polarity towards Bob when $\PB > \PA$. 
	Given a polarised protocol $\mathcal{R}$, we may switch the roles of Alice and Bob since the definition of coin flipping is symmetric. 
	To make the polarity explicit, we define $\mathcal{R}_A$ to be the version of the protocol with $\PA > \PB$ and $\mathcal{R}_B$ to be the version with $\PB > \PA$.}
	With this in mind, we can now define a simple composition. 
	
	\begin{lyxalgorithm}[Winner-gets-polarity composition] \label{alg:simple}  
	Alice and Bob agree on a protocol $\mathcal{R}$. 
	\begin{enumerate} 
	\item Alice and Bob execute protocol $\mathcal{R}$. 
	\item If Alice wins, she polarises the second protocol towards herself, i.e., they now use the protocol $\mathcal{R}_A$ to determine the final outcome. %
	\item If Bob wins, he polarises the second protocol towards himself, i.e., they now use the protocol $\mathcal{R}_B$ to determine the final outcome. %
	\end{enumerate} 
	\end{lyxalgorithm}  

	The standard composition above is a sensible way to balance the cheating probabilities of a protocol. 
	For instance, if $\mathcal{R}$ has cheating probabilities $\PA$ and $\PB$ with $\PA > \PB$, then the composition gets to decide ``who gets to be Alice'' in the second run.  
	We can easily compute Alice's cheating probability in the composition as 
	\begin{equation} \label{first}
	(\PA)^2 + (1-\PA) \PB < \PA 
	\end{equation}  
	and Bob's as 
	\begin{equation} \label{second}
	\PB \PA + (1- \PB) \PB < \PA. 
	\end{equation} 
	This does indeed reduce the bias since the maximum cheating probability is now smaller.  

	\paragraph{Abort-phobic composition.} 
	The ``traditional'' way of considering WCF protocols is to view them as only having two outcomes ``Alice wins'' (when $c = 0$) or ``Bob wins'' ($c = 1$). 
	This is because Alice can declare herself the winner if she catches Bob cheating. 
	Similarly, Bob can declare himself the winner if he catches Alice cheating.\footnote{In doing so, we implicitly assume that the protocol has perfect correctness---when both players are honest, the probability of abort is zero.} 
	This is completely fine when we consider ``one-shot'' versions of these protocols, but we lose something when we compose them. 
	For instance, in the simple composition used in ~\Algref{simple}, Bob should not really accept to continue onto the second protocol if he catches Alice cheating in the first. 
	That is, if he knows Alice cheated, he can declare himself the winner of the entire protocol.  
	In other words, the cheating probabilities~(\ref{first}) and (\ref{second}) may get reduced even further.  
	For purposes of this discussion, suppose Bob adopts a cheating strategy which has a probability $v_B$ of him winning ($c = 1$), a probability $v_A$ of him losing ($c = 0$), and a probability $v_{\perp}$ of Alice catching him  cheating. 
	Then his cheating probability in the (abort-phobic) version of the simple composition is now 
	\begin{equation} 
	v_B \cdot \PA \, + \, v_A \cdot \PB \, + \, v_{\perp} \cdot 0. 
	\end{equation} 
	This quantity may be a strict improvement if $v_{\perp} > 0$ when $v_B = \PB$.  

	The concept of abort-phobic composition is simple. 
	Alice and Bob keep using WCF protocols and the winner (at that round) gets to choose the polarity of the subsequent protocol. 
	However, if either party \emph{ever aborts}, then it is game over and the cheating player loses \emph{the entire composite protocol}. 

	One may think it is tricky to analyse abort-phobic compositions, but we may do this one step at time. 
	To this end, we introduce the concept of \emph{cheat vectors}. 

	\begin{defn}[$\mathbb{C}_A,\mathbb{C}_B$; Alice and Bob's cheat vectors]  
	\label{def:CheatVectors}
	Given a protocol $\mathcal{R}$, we say that $(v_A, v_B, v_{\perp})$ is a cheat vector for (dishonest) Bob if there exists a cheating strategy where,
	\begin{center} 
	\begin{tabularx}{\textwidth}{lX}
		$v_B$ \, is the probability with which Alice accepts the outcome $c = 1$, \\ 
		$v_A$ \, is the probability with which Alice accepts the outcome $c = 0$, \\ 
		$v_{\perp}$ \, is the probability with which Alice aborts. \\ 
	\end{tabularx} 
	\end{center}  
	We denote the set of cheat vectors for (dishonest) Bob by $\mathbb{C}_B({\mathcal{R})}$. Cheat vectors for (dishonest) Alice and $\mathbb{C}_A({\mathcal{R})}$ are analogously defined keeping the notation  $v_A$ for her winning, $v_B$ for her losing, and $v_{\perp}$ for Bob aborting. 
	\end{defn} 
	
	In this work, we show how to capture cheat vectors as the feasible region of a semi-definite program, from which we can optimize 
	\begin{equation} 
	v_B \cdot \PA \, + \, v_A \cdot \PB \, + \, v_{\perp} \cdot 0. 
	\end{equation}
	For this to work, we assume we have $\PA$ and $\PB$ for the protocol that comes in the second round. 
	A simplifying observation is that once we solve for the optimal cheating probabilities in the abort-phobic composition in this way, we can then fix those probabilities and compose again.	In other words, we are recursively composing the abort-phobic composition, from the \emph{bottom up}. 

	By using abort-phobic compositions with \Algref{AliceSelfTests} (where Alice self-tests) one obtains protocols which converge onto a bias of $\varepsilon \approx 0.31486$ proving the first part of the main result. For the second part, we place \Algref{AliceSelfTests} at the bottom, and \Algref{BobSelfTests} (where Bob self-tests) on higher layers, to obtain protocols whose bias approaches $\varepsilon \approx 0.29104$. %
	These results are also contingent on the assumption that \Conjref{Pcont,Qcont} hold. 

	\subsection{Applications} 

	The concept of polarity extends beyond finding WCF protocols and, as such, the ``winner-gets-polarity'' concept allows for WCF to be used in other compositions. 
	Indeed, we can use it to balance the cheating probabilities in  \emph{any} polarised protocol for any symmetric two-party cryptographic task for which such notions can be properly defined. 
	
	For instance, many \emph{strong} coin flipping protocols can be thought of as polarised. 
	For an example, ~\Algref{SCF} is indeed a polarised strong coin flipping protocol. 
	Thus, by balancing the cheating probabilities of that protocol using our DI WCF protocol, 
	we get the following corollary. 

	\begin{cor} 
	Suppose \Conjref{Pcont,Qcont} hold. Then, there exist DI strong coin flipping protocols where no party can cheat with probability greater than $0.33192$. 
	\end{cor} 
To contrast, for \cite{Silman2011}, the bound on cheating probabilities was $0.336637$. There are likely more examples of protocols which can be balanced in a DI way using this idea.  
		  
\section{New protocols using self-testing | First Technique}
\label{sec:FirstTechSelfTest}
We start by recalling the DI strong coin flipping protocol introduced in \cite{Silman2011}, \Algref{SCF}, and introduce its weak coin flipping variant \Algref{WCF}. We then describe the new \Algref{AliceSelfTests,BobSelfTests}, where Alice and Bob respectively perform the self-testing step. We also give more formal security guarantees associated with these. Their proofs constitute \Secref{securityAsymptotic,SecurityFiniteN}. 

\paragraph{Notation} %
We often use single calligraphic symbols $\mathcal S, \mathcal W, \mathcal P$ and $\mathcal Q$ to succinctly refer to the aforementioned protocols. We use, for instance, $p^*_A(\cal{W})$ (resp. $p^*_B(\cal{W})$) to denote the maximum probability with which a dishonest Alice (resp. Bob) can force an honest Bob (resp. Alice) to output ``Alice'' (resp. ``Bob'') in an execution of $\cal{W}$. When we say, for instance, consider a tripartite device, viz. a triple of boxes $\Box^A,\Box^B,\Box^C$, we mean that there is a tripartite quantum state and local measurements associated with these boxes. The input to the box selects the measurement setting and the output is the measurement outcome as governed by quantum theory (see \Defref{box}). When we speak of Alice and Bob exchanging boxes, we understand that the these states and the description of the measurement settings are sent over a (possibly insecure) quantum communication channel (see \Defref{BoxProtocol,MEprotocol} in \Secref{BoxParadigm}). 

We recall the GHZ test before starting our main discussion as this is at the heart of these protocols.

\begin{defn}
\label{def:GHZ-box}Suppose we are given a tripartite device, viz. a triple of boxes, $\Box^{A},\Box^{B}$
and $\Box^{C}$, which accept binary inputs $a,b,c\in\{0,1\}$ and
produces binary output $x,y,z\in\{0,1\}$ respectively. The boxes
pass the GHZ test if $a\oplus b\oplus c=xyz\oplus1$, given the inputs
satisfy $x\oplus y\oplus z=1$.
\end{defn}
 
It is known that no classical triple of boxes can pass the GHZ test with certainty but quantum boxes can.

\begin{claim}
\label{claim:Quantum-boxes-pass}Quantum boxes pass the GHZ test with
certainty (even if they cannot communicate), for the state $\left|\psi\right\rangle _{ABC}=\frac{\left|000\right\rangle _{ABC}+\left|111\right\rangle _{ABC}}{\sqrt{2}}$,
and measurement\footnote{we added the identity so that the eigenvalues associated become $0,1$
instead of $-1,1$.} $\frac{\sigma_{x}+\mathbb{I}}{2}$ for input $0$ and $\frac{\sigma_{y}+\mathbb{I}}{2}$
for input $1$ (in the notation introduced earlier, $M_{0|0}^{A}=\left|+\right\rangle \left\langle +\right|,M_{1|0}^{A}=\left|-\right\rangle \left\langle -\right|$
and so on, where $\left|\pm\right\rangle =\frac{\left|0\right\rangle \pm\left|1\right\rangle }{\sqrt{2}}$).
\end{claim}

The proof is easier to see in the case where the outcomes are $\pm1$;
it follows from the observations that $\sigma_{y}\otimes\sigma_{y}\otimes\sigma_{y}\left|\psi\right\rangle =-\left|\psi\right\rangle $,
$\sigma_{x}\otimes\sigma_{x}\otimes\sigma_{x}\left|\psi\right\rangle =\left|\psi\right\rangle $
and the anti-commutation of $\sigma_{x}$ and $\sigma_{y}$ matrices,
i.e. $\sigma_{x}\sigma_{y}+\sigma_{y}\sigma_{x}=0$.

In fact a stronger property holds. If a triple of boxes passes the GHZ test with certainty, it can be shown that up to a local isometry, the state and measurements are as in \Claimref{Quantum-boxes-pass} above. %

\begin{lem}\label{lem:rigidityGHZ}
	Let $a,b,c,x,y,z\in\{0,1\}$. Consider a triple of quantum boxes, specified
	by projectors\footnote{This is without loss of generality; given POVMs and a state, one can always construct projectors and a state on a larger Hilbert space which preserves the statistics, using Naimark's theorem.} $\{M_{a|x}^{A},M_{b|y}^{B},M_{c|z}^{C}\}$ acting on
	finite dimensional Hilbert spaces $\mathcal{H}^{A},\mathcal{H}^{B}$
	and $\mathcal{H}^{C}$, and $\left|\psi\right\rangle \in\mathcal{H}^{A}\otimes\mathcal{H}^{B}\otimes\mathcal{H}^{C}=:\mathcal{H}^{ABC}$.
	If the triple pass the GHZ test with probability $1-\epsilon$ (for
	$0\le \epsilon < 1$), then there exists a local isometry, 
	\[
	\Phi=\Phi^{A}\otimes\Phi^{B}\otimes\Phi^{C}:\mathcal{H}^{ABC}\to\mathcal{H}^{ABC}\otimes\mathbb{C}^{2\times3}
	\]
	and a decreasing function of $\epsilon$, $f(\epsilon)$ such that
	\begin{align}
		\label{eq:continuity_GHZ}
	\left\Vert \Phi\left(\left|\psi\right\rangle \right)-\left|\chi\right\rangle \otimes\left|{\rm junk}\right\rangle \right\Vert  & \le f(\epsilon),\\
	\left\Vert \Phi\left(M_{d|t}^{D}\left|\psi\right\rangle \right)-\Pi_{d|t}^{D}\left|{\rm GHZ}\right\rangle \otimes\left|{\rm junk}\right\rangle \right\Vert  & \le f(\epsilon)\quad\forall D\in\{A,B,C\},\text{ and }d,t\in\{0,1\}
	\end{align}
	where $\left|{\rm GHZ}\right\rangle =\frac{\left|000\right\rangle +\left|111\right\rangle }{\sqrt{2}}\in\mathbb{C}^{2\times3}$,
	$\left|{\rm junk}\right\rangle \in\mathcal{H}^{ABC}$ is some arbitrary
	state and $\{\Pi_{a|x}^{A},\Pi_{b|y}^{B},\Pi_{c|z}^{C}\}$ are projectors
	corresponding to $\sigma_{x}$ on the first, second and third qubit
	of $\left|{\rm GHZ}\right\rangle $ respectively, for $x=0$ and corresponding
	to $\sigma_{y}$ for $x=1$, as in \Claimref{Quantum-boxes-pass}.
\end{lem}
\begin{proof}
	Proofs of robust self-testing for GHZ can be found in \cite{MillerShi} and \cite{McKague}.
\end{proof}
  
\subsection{Original protocols}
\label{subsec:SCForiginal}

\Algref{SCF} is defined as follows. 
\begin{varalgorithm}{S} 
	\caption{\quad A DI-SCF protocol with $\PA = \cos^2{\pi/8}$ and $\PB = 3/4$ (\cite{Silman2011})} 
	\label{alg:SCF}  

    Alice has one box and Bob has two boxes. 
    Each box takes one binary input and gives one binary output and are designed to play the optimal GHZ game strategy. 
    (Who creates and distributes the boxes is not important in the DI setting.) 
    \begin{enumerate}
        \item Alice chooses a uniformly random input to her box $x \in_R \{ 0, 1 \}$ and obtains the outcome $a$. 
        She chooses another uniformly random bit $r \in_R \{ 0, 1 \}$ and computes $s = a \oplus (x \cdot r)$. 
        She sends $s$ to Bob. 
    
        \item Bob chooses a uniformly random bit $g \in_R \{ 0, 1 \}$ and sends it to Alice. 
        (We may think of $g$ as Bob's ``guess'' for the value of $x$.) 
    
        \item Alice sends $x$ to Bob. 
        They both compute the output $c = x \oplus g$. 
        (This is the outcome of the protocol if no-one abort.)
        
        \item Bob tests Alice
        \begin{enumerate} 
            \item[\textup{Test 1}:] Alice sends $a$ to Bob. Bob sees if $s = a$ or $s = a \oplus x$.                If this is not the case, he aborts. 
            \item[\textup{Test 2}:] Bob chooses $y,z\in_{R}\{0,1\}$ uniformly at random such that $x\oplus y\oplus z=1$ and then performs a GHZ using $x,y,z$ as the inputs and $a,b,c$ as the output from the three boxes. He aborts if this test fails.
        \end{enumerate} 
        
    \item If Bob does not abort, they both accept the value of $c$ as the outcome of the protocol. 
    \end{enumerate} 
\end{varalgorithm} 

We now discuss the correctness and soundness of \Algref{SCF}. From \Claimref{Quantum-boxes-pass}, it is clear that when both players
follow the protocol using GHZ boxes (\Defref{GHZ-box}),
Bob never aborts and they win with equal probabilities. As for the security, \cite{Silman2011} proved the following.

\begin{lem}[Security of SCF]
 \cite{Silman2011} Let $\mathcal{S}$ denote the protocol corresponding
to \Algref{SCF}. Then, the success probability of cheating
Bob,\footnote{For SCF, $P^*_B$ is max\{Pr[Bob can force Alice to output 1], Pr[Bob can force Alice to output 0]\}; $P^*_A$ is analogously defined.} $p_{B}^{*}(\mathcal{S})\le\frac{3}{4}$ and that of cheating
Alice, $p_{A}^{*}(\mathcal{S})\le\cos^{2}(\pi/8)$. 

Further, both
bounds are saturated by a quantum strategy which uses a GHZ state
and the honest player measures along the $\sigma_{x}/\sigma_{y}$
basis corresponding to input $0/1$ into the box. Cheating Alice measures
along $\sigma_{\hat{n}}$ for $\hat{n}=\frac{1}{\sqrt{2}}(\hat{x}+\hat{y})$
while cheating Bob measures his first box along $\sigma_{x}$ and
second along $\sigma_{y}$. \label{lem:SCFstandard}
\end{lem}
 
{Note that both players can cheat maximally assuming they share a GHZ
state and the honest player measures along the associated basis. This is why it was asserted 
that even though the cheating player could potentially tamper
with the boxes before handing them to the honest player,
exploiting this freedom does not offer any advantage to the cheating
player. 

Clearly, if we take \Algref{SCF} as is and treat it like a weak coin flipping protocol, this conclusion would continue to hold. As motivated in the introduction, we consider a minor, yet crucial, modification to \Algref{SCF}. Observe that in \Algref{SCF} only Bob performs the test round,
while in weak coin flipping there is a notion of Alice winning and Bob winning which may be leveraged. More precisely,
if $x\oplus g=0$, i.e. the outcome corresponding to ``Alice wins'',
we can imagine that Bob continues to perform the test to ensure (at
least to some extend) that Alice did not cheat. However, if $x\oplus g=1$,
i.e. the outcome corresponding to ``Bob wins'', we can require Alice
to now complete the GHZ test to ensure that Bob did not cheat. Since we analyse this protocol in detail, we state it as \Algref{WCF}, somewhat redundantly below. We have emphasised the changes compared to \Algref{SCF} in {\color{blue} blue}.

\begin{varalgorithm}{W}
	\caption{\quad Weak Coin Flipping version of \Algref{SCF} (Italics indicate the differences with \Algref{SCF})} 
	\label{alg:WCF} 
	Alice has one box and Bob has two boxes. 
	Each box takes one binary input and gives one binary output and are designed to play the optimal GHZ game strategy. 
	(Who creates and distributes the boxes is not important in the DI setting.) 
	\begin{enumerate}
		
	    \item Alice chooses a uniformly random input to her box $x \in_R \{ 0, 1 \}$ and obtains the outcome $a$. 
	    She chooses another uniformly random bit $r \in_R \{ 0, 1 \}$ and computes $s = a \oplus (x \cdot r)$.
	    She sends $s$ to Bob.
	    
	    \item Bob chooses a uniformly random bit $g \in_R \{ 0, 1 \}$ and sends it to Alice. 
	    (We may think of $g$ as Bob's ``guess'' for the value of $x$.) 
	    
	    \item Alice sends $x$ to Bob. 
	    They both compute the output $c = x \oplus g$. 
	    This is the outcome of the protocol assuming neither Alice nor Bob aborts. 

	    \item Test rounds:
	    \begin{enumerate}
	        \item {\color{blue}\emph{If $x\oplus g=0$,}} Bob tests Alice
            \begin{enumerate} 
                \item[\textup{Test 1}:] Alice sends $a$ to Bob. Bob sees if $s = a$ or $s = a \oplus x$.                If this is not the case, he aborts. 
                \item[\textup{Test 2}:] Bob chooses $y,z\in_R\{0,1\}$ uniformly at random such that $x\oplus y\oplus z=1$ and then performs a GHZ using $x,y,z$ as the inputs and $a,b,c$ as the output from the three boxes. He aborts if this test fails.
            \end{enumerate} 
	
			{\color{blue} 
	        \item \emph{If $x\oplus g=1$, Alice tests Bob}
	        \begin{enumerate}
	            \item[\textup{\emph{Test 3:}}] \emph{Alice chooses $y,z\in_{R}\{0,1\}$  uniformly at random  such that $x\oplus y\oplus z=1$ and sends them to Bob. Bob inputs $y,z$ into his boxes, obtains and sends $b,c$ to Alice. Alice tests if $x,y,z$ as inputs and $a,b,c$ as outputs, satisfy	the GHZ test. She aborts if this test fails.}
            \end{enumerate}
			}
        \end{enumerate}
        
	    \item If Alice and Bob do not abort, they both accept the value of $c$ as the outcome of the protocol.
	    
    \end{enumerate} 
\end{varalgorithm}

While it is not surprising that $p_A^*(\mathcal{W})=p_A^*(\mathcal{P})=\cos^2(\pi/8)$, it turns out that $p_{B}^{*}(\mathcal{W})=p_B^*(\mathcal{P})=3/4$, despite the additional test that Alice performs i.e. $P^*_B$ for \Algref{WCF} is not lowered. Yet, this is not quite a setback---one can show that the best cheating strategy now deviates from the GHZ state and measurements for the honest player, suggesting that a cheating player \emph{does} benefit from tampering with the boxes. Consequently, adding a self-testing step before initiating \Algref{WCF}, may potentially improve its security and as we shall see in the following subsections, it indeed does.

A remark about the limitation of self-testing in this setting. We note that no self-testing scheme can be concocted which simultaneously self-tests Alice and Bob's boxes. More precisely, no such procedure can ensure that Alice and Bob share a GHZ state (Alice one part, Bob the other two, for instance) because this would mean perfect (or near perfect) SCF is possible which,\footnote{More precisely, note that once a GHZ state has been shared, and we are in the device dependent setting, then the protocol would be for each player to make a projective $\sigma_z$ measurement. What the malicious player does becomes irrelevant for the security.} recall, is forbidden even in the device dependent case.\footnote{More precisely, Kitaev \cite{Kitaev03} showed that for any SCF protocol, $\epsilon\ge\frac{1}{\sqrt{2}}-\frac{1}{2}$. Note that the protocol in this paper does not violate the bound---we propose coin flipping protocols with not-too-small cheating probabilities.} {In fact, \cite{bansal2023impossibility} shows that the impossibility extends to all devices that produce non-product correlations.}

}

\subsection{Alice self-tests | Protocol $\mathcal{P}$} 
We begin by explicitly stating \Algref{AliceSelfTests} (where Alice self-tests her boxes before initiating \Algref{WCF}). In the honest implementation, the quantum device---the triple of boxes---used in \Algref{AliceSelfTests} are characterised by the GHZ state and measurements (see \Claimref{Quantum-boxes-pass}).

\begin{varalgorithm}{P}
	\caption{\quad Alice self-tests} \label{alg:AliceSelfTests}  

	Alice starts with $n$ boxes, indexed from $\rom{1}_1$ to $\rom{1}_n$. 
	Bob starts with $2n$ boxes, the first half indexed by $\rom{2}_1$ to $\rom{2}_n$ and the last half indexed by $\rom{3}_1$ to $\rom{3}_n$. 
	The triple of boxes $(\rom{1}_i, \rom{2}_i, \rom{3}_i)$ is meant to play the optimal GHZ game strategy.   
	\begin{enumerate}    
	\item Alice selects a uniformly random index $i \in \{ 1, \ldots, n \}$ and asks Bob to send her all the boxes \emph{except} those indexed by $\rom{2}_i$ and $\rom{3}_i$. 
	\item Alice performs $n-1$ GHZ tests using the $n-1$ triples of boxes she has, making sure there is no communication between any of them, e.g. by shielding the boxes (in the relativistic settings, coin flipping is possible). 
	\item Alice aborts if \emph{any} of the GHZ tests fail. 
	Otherwise, she announces to Bob that they can use the remaining boxes for Protocol~\ref{alg:WCF}.  
	\end{enumerate} 
\end{varalgorithm}

{
\begin{prop}
    \label{prop:AliceSelfTests} Let $\mathcal{P}$ denote \Algref{AliceSelfTests} and suppose \Conjref{Pcont} holds. Then, in the $n\to\infty$ limit (i.e. very large number of devices are used in the self-test step), Alice's cheating probability $p_{A}^{*}(\mathcal{P})\le\cos^{2}(\pi/8)\approx0.852$ and Bob's cheating probability $p_{B}^{*}(\mathcal{P}) \le 0.667$. 
\end{prop}
}

{We defer the proof to \Subsecref{SDP-when-Alice} and \Subsecref{SDP-when-Alice-N}. As remarked in the introduction, the value for $p_{B}^{*}(\mathcal{P})$ is lower than $p^*_B(\mathcal{W})$ and was obtained by numerically solving the corresponding SDP while the analysis for cheating Alice is the same as that of the original protocol. }

{
To write the associated security statement that we use later for compositions, we make the following assumption. (We also state the corresponding statement for \Algref{BobSelfTests}, i.e. when Bob self-tests.)
\begin{assumption}\label{assu:GHZwinwithepsilon}
    In protocol $\mathcal{P}$ ($\mathcal{Q}$, resp.),
Alice (Bob, resp.) does not perform the box verification step and instead
it is assumed that her box is (his boxes are, resp.) taken from a triple of
boxes which win the GHZ game with probability $1-\epsilon$. Here $\eps$ is a decreasing function of $n$ satisfying $\lim_{n\to \infty} \epsilon=0$. 
\end{assumption}
Recalling the definition of cheat vectors from the introduction (see \Defref{CheatVectors}), and using \Assuref{GHZwinwithepsilon} above instead of the continuity conjecture, \Conjref{Pcont}, we have the following:
} 

\begin{lem}
\label{lem:AliceSelfTests}Let $\mathcal{P}$ denote \Algref{AliceSelfTests} and suppose \Assuref{GHZwinwithepsilon} holds. Further, let $c_{0},c_{1},c_{\perp}\in\mathbb{R}$, and $\mathbb{C}_{B}(\mathcal{P})$ be the set of cheat vectors for Bob (see \Figref{cheatVectors_ProtocolP}). 
Then, as $n\to\infty$, the solution
to the optimisation problem $\max(c_{0}\alpha+c_{1}\beta+c_{\perp}\gamma)$
over $\mathbb{C}_{B}(\mathcal{P})$ approaches that of an SDP over variables of constant dimension (wrt $n$). (In particular, i.e. for $c_{0}=c_{\perp}=0$ and $c_{1}=1$,
$p_{B}^{*}(\mathcal{P}) \approx 0.667$.)
\end{lem}

{The fact that optimising linear functions in Bob's cheat vectors is an SDP becomes useful in \Secref{Second-Technique} when we compose these protocols. Again, the proofs are deferred to \Subsecref{SDP-when-Alice,SDP-when-Alice-N}.}

\subsection{Bob self-tests | Protocol $\mathcal{Q}$}

We analogously define \Algref{BobSelfTests}---where Bob self-tests his boxes before initiating \Algref{WCF}. 

\begin{varalgorithm}{Q}
	\caption{\quad Bob self-tests}
\label{alg:BobSelfTests}

	Alice starts with $n$ boxes, indexed from $\rom{1}_1$ to $\rom{1}_n$. 
	Bob starts with $2n$ boxes, the first half indexed by $\rom{2}_1$ to $\rom{2}_n$ and the last half indexed by $\rom{3}_1$ to $\rom{3}_n$. 
	The triple of boxes $(\rom{1}_i, \rom{2}_i, \rom{3}_i)$ is meant to play the optimal GHZ game strategy.   
	\begin{enumerate}    
	\item Bob selects a uniformly random index $i \in \{ 1, \ldots, n \}$ and asks Alice to send him all the boxes \emph{except} those indexed by $\rom{1}_i$. 
	\item Bob performs $n-1$ GHZ tests using the $n-1$ triples of boxes he has, making sure there is no communication between any of them. 
	\item Bob aborts if \emph{any} of the GHZ tests fail. 
	Otherwise, he announces to Alice that they can use the remaining boxes for Protocol~\ref{alg:WCF}.  
	\end{enumerate} 

\end{varalgorithm} 

Consider \Algref{WCF} and \Algref{SCF}. Suppose Bob is honest while Alice is malicious, and that at step 3, she sends an $x$ s.t. $x\oplus g = 0$. Under these conditions, observe that Bob's actions are identical in both \Algref{WCF} and \Algref{SCF}. Since it is already known from \Lemref{SCFstandard} that Alice does not gain anything from tampering with Bob's boxes, the same conclusion holds for \Algref{WCF}. Thus, for \Algref{BobSelfTests}, we do not expect any improvement in Bob's security, viz. $p^*_A(\mathcal{Q})=p^*_A(\mathcal{W})$ given that \Algref{BobSelfTests} only ensures Alice does not tamper with Bob's boxes. It is also immediate that $p^*_B(\mathcal{Q})=p^*_B(\mathcal{W})$. 
{This means that we do not see any advantage of self-testing at this stage and therefore the analogue of \Propref{AliceSelfTests} is not stated. However, self-testing does help when compositions are considered. More precisely, analogously to \Algref{AliceSelfTests}, optimisation of linear functions of Alice's cheat vectors now becomes an SDP and we reap the benefits of this simplification in the next section. }

\begin{lem}
\label{lem:BobSelfTests} 
Let $\mathcal{Q}$ denote \Algref{BobSelfTests} and suppose \Assuref{GHZwinwithepsilon} holds. %
Further, let $c_{0},c_{1},c_{\perp}\in\mathbb{R}$, and $\mathbb{C}_{A}(\mathcal{Q})$
be the set of cheat vectors for Alice (see \Figref{cheatVectors_ProtocolQ}). Then, as $n\to\infty$, the
solution to the optimisation problem $\max(c_{0}\alpha+c_{1}\beta+c_{\perp}\gamma)$
over $(\alpha,\beta,\gamma)\in\mathbb{C}_{A}(\mathcal{Q})$ approaches
that of an SDP of constant dimension (wrt $n$). 
\end{lem}

\begin{figure}	
	\begin{centering}
		\subfloat[Cheat vectors $\mathbb{C}_B(\mathcal{P})$. The $x$-axis represents $v_B$ and the $y$-axis represents the smallest $v_{\perp}$, given $v_B$.\label{fig:cheatVectors_ProtocolP}]{
				\begin{centering}
				\includegraphics[width=7cm]{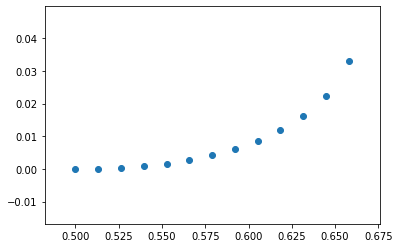}
				
				\end{centering}}
	\quad 
		\subfloat[Cheat vectors $\mathbb{C}_A(\mathcal{Q})$. The $x$-axis represents $v_B$ and the $y$-axis represents the smallest $v_{\perp}$, given $v_B$. \label{fig:cheatVectors_ProtocolQ}]{
			\begin{centering}
			
			\includegraphics[width=7cm]{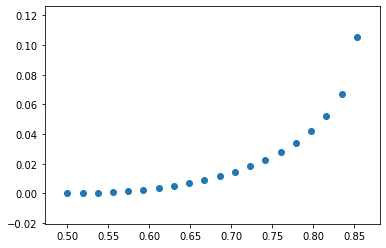}
			\end{centering}}
	\par
	\end{centering}
\caption{Cheat vectors for \Algref{AliceSelfTests} and \Algref{BobSelfTests}. Observe how compared to \Figref{cheatVectors_ProtocolP}, the abort probabilities in \Figref{cheatVectors_ProtocolQ} are higher making it more suitable for abort-phobic compositions.}	
\end{figure} 

The proof is again deferred to \Subsecref{SDP-when-Bob} and \Subsecref{SDP-when-Bob-N}.

\section{New protocols using abort-phobic compositions | Second Technique}\label{sec:Second-Technique}

{To convey the key idea, here also we assume \Assuref{GHZwinwithepsilon} except with $\epsilon=0$, i.e. 
the player that does the self-test, in fact, holds boxes corresponding exactly to a GHZ strategy. This simplification allows us to focus on the novelty in our composition technique. In \Secref{SecurityFiniteN} we complete the analysis by accounting for all the details skipped here.}

Central to the discussion in this section, will be the notion of \emph{polarity} introduced in \Subsecref{IntroSecondTechnique} and our results will apply to polarised protocols. Note that $\mathcal{W},\mathcal{P},$ and $\mathcal{Q}$ are all polarised. We begin by recalling that in \Subsecref{IntroSecondTechnique}, again, we had introduced a ``standard composition''---the simplest implementation of the ``winner gets polarity'' idea. Here, we restate this composition with more precision and introduce the notation we use for the more involved cases.

\subsection{Composition}
The {definition below} simply formalises the following---execute protocol $\mathcal{X}$ to determine who gets to choose the polarity of protocol $\mathcal{Y}$. We use $C$ with two parameters, as in $C(\mathcal{X},\mathcal{Y})$, to denote a single composition described above. We use $C(\mathcal{X})$ to denote repeated compositions of $\mathcal{X}$.
\begin{defn}[$C(\cdot,\cdot)$ and $C(\cdot)$]
\label{def:C}
Given two polarised WCF protocols, $\mathcal{X}$
and $\mathcal{Y}$, let $\mathcal{X}_{A},\mathcal{X}_{B}$ and $\mathcal{Y}_{A},\mathcal{Y}_{B}$
be their polarisations (see \Subsecref{IntroSecondTechnique}). 
Define
$C(\mathcal{X},\mathcal{Y})$ as follows:
\begin{enumerate}
\item Alice and Bob execute $\mathcal{X}_{A}$ and obtain outcome $X\in\{A,B,\perp\}$. 
\item 
\begin{enumerate}
\item If $X=A$, execute $\mathcal{Y}_{A}$ and obtain outcome $Y\in\{A,B,\perp\}$,
else 
\item if $X=B$, execute $\mathcal{Y}_{B}$ and obtain outcome $Y\in\{A,B,\perp\}$,
and finally
\item if $X=\perp$, set $Y=\perp$.
\end{enumerate}
Output $Y$.
\end{enumerate}
Let $\mathcal{Z}^{i+1}:=C(\mathcal{X},\mathcal{Z}^{i})$ for $i\ge1$,
and $\mathcal{Z}^{1}:=\mathcal{X}$. Then, formally, define $C(\mathcal{X}):=\lim_{i\to\infty}\mathcal{Z}^{i}$.\footnote{This formal notation is defined to let $p_A^*(C(\mathcal{X}))$ mean $\lim_{i\to \infty} p_A^*(\mathcal{Z}^i)$ and similarly for $p_B^*$ and $\epsilon$.} 
\end{defn}

\branchcolor{black}{The study of such composed protocols is simplified by assuming that
in an honest run, neither player outputs $\perp$ (abort), i.e. they
either output $A$ or $B$. We take a moment to explain this.

Consider any protocol $\mathcal{R}$ where, when both players are
honest, the probability of abort is zero. The protocols we have described
so far, satisfy this property, so long as we assume that honest players
can prepare perfect GHZ boxes. Such protocols are readily converted
into protocols where an honest player never outputs abort. 

For instance, suppose that in the execution of the aforementioned
protocol $\mathcal{R}$ (with no-honest-abort), Alice behaves honestly
but Bob is malicious. Suppose after interacting with Bob, Alice reaches
the conclusion that she must abort. Since she knows that if Bob was
honest, the outcome abort could not have arisen, she concludes that
Bob is cheating and declares herself the winner, i.e. she outputs
$A$. Similarly, when Bob is honest and after the interaction, reaches
the outcome abort, he knows Alice cheated and can therefore declare
himself the winner, i.e. output $B$. 

Whenever we modify a protocol so that an honest Alice (Bob) replaces
the outcome abort with Alice (Bob) winning, we say Alice (Bob) is
\emph{lenient}. This is motivated by the fact that when we compose
protocols, if Alice can conclude that Bob is cheating, and she still
outputs $A$ instead of aborting, she is giving Bob further opportunity
to cheat---she is being lenient.}
\begin{defn}[$\mathcal{R}$ with lenient players]
\label{def:lenientR} Suppose $\mathcal{R}$ is a WCF protocol such
that when both players are honest, the probability of outcome abort,
$\perp$, is zero. Then by ``\emph{$\mathcal{R}$ with lenient Alice
(Bob)}'', which we denote by $\mathcal{R}^{L\perp}$ ($\mathcal{R}^{\perp L}$),
we mean that Alice (Bob) follows $\mathcal{R}$ except that the outcome
$\perp$ replaced with $A$ ($B$). Finally, by ``\emph{lenient $\mathcal{R}$}'',
which we denote by $\mathcal{R}^{LL}$, we mean $\mathcal{R}$ with
lenient Alice and Bob. 
\end{defn}

\branchcolor{black}{For clarity and conciseness, we define $C^{LL}$ to be compositions
with lenient variants of the given protocols. We work out some examples
of such protocols and analyse their security in the following section.
These can be improved by considering $C^{L\perp}$ and $C^{\perp L}$---compositions
where only one player is lenient. We discuss those afterwards.}
\begin{defn}[$C^{LL}$, $C^{\perp L}$ and $C^{L\perp}$]
 Suppose a WCF protocol $\mathcal{X}$ can be transformed into its
\emph{lenient} variants (see \Defref{lenientR}). Then define 
\begin{align*}
C^{LL}(\mathcal{X},\mathcal{Y}) & :=C(\mathcal{X}^{LL},\mathcal{Y}),\\
C^{\perp L}(\mathcal{X},\mathcal{Y}) & :=C(\mathcal{X}^{\perp L},\mathcal{Y}),\quad\text{and}\\
C^{L\perp}(\mathcal{X},\mathcal{Y}) & :=C(\mathcal{X}^{L\perp},\mathcal{Y}).
\end{align*}
In words, $C^{LL}$ is referred to as a \emph{standard }composition,
while $C^{\perp L}$ and $C^{L\perp}$ are referred to as \emph{abort-phobic}
compositions. 
The single argument versions are analogously defined, i.e. $C^{LL}(\mathcal{X}):=C(\mathcal{X}^{LL})$, $C^{L\perp}(\mathcal{X}):=C(\mathcal{X}^{L\perp})$ and $C^{\perp L}(\mathcal{X}^{\perp L})$.
\end{defn}

We make two remarks. \emph{First}, the reader might wonder why do we not consider $C^{\perp\perp}$. Nothing prevents us from doing this, however, to obtain an improvement, one must also analyse such a composition. This is, in general, difficult in the device independent setting since one must optimise over both states and measurement settings. However, due to the self-testing step, the (asymptotic) analysis of $C^{\perp L}$ and $C^{L \perp}$ reduces to a semi-definite program which can be solved for protocols with a small number of rounds.
\emph{Second}, it is worth clarifying that for our analysis, we need not consider Alice aborting and Bob aborting as separate events. When both players are honest, correctness requires both players produce the same output. When one player is malicious and the other honest, we simply assume that the malicious player can learn the honest player's outcome before producing their outcome. Therefore, it suffices for the malicious player to only care about the output produced by the honest player. This aforementioned assumption only makes the malicious player stronger and therefore can be made in the security analysis.

\subsection{Standard Composition | $C^{LL}$}

\branchcolor{black}{We begin with the simplest case, standard composition, $C^{LL}$.
Let us take an example. Let $\mathcal{P}$ denote \Algref{AliceSelfTests}
and recall (see \Lemref{AliceSelfTests})
\begin{align*}
p_{A}^{*}(\mathcal{P}_{A}) & =:\alpha\approx0.852,\\
p_{B}^{*}(\mathcal{P}_{A}) & =:\beta\approx0.667.
\end{align*}
Note that therefore $p_{A}^{*}(\mathcal{P}_{B})=\beta$ and $p_{B}^{*}(\mathcal{P}_{B})=\alpha$.
Further, let $\mathcal{P}':=C^{LL}(\mathcal{P},\mathcal{P})$, i.e.
Alice and Bob (who are both lenient) first execute $\mathcal{P}_{A}$
and if the outcome is $A$, they execute $\mathcal{P}_{A}$, while
if the outcome is $B$, they execute $\mathcal{P}_{B}$. This is illustrated
in \Figref{Standard-composition-technique} where note that the event
abort doesn't appear due to the leniency assumption. This allows us
to evaluate the cheating probabilities for the resulting protocol
as 
\begin{align}
p_{A}^{*}(\mathcal{P}') & =\alpha\alpha+(1-\alpha)\beta=:\alpha^{(1)},\quad\text{and}\label{eq:pStarAPprimeA}\\
p_{B}^{*}(\mathcal{P}') & =\beta\alpha+(1-\beta)\beta=:\beta^{(1)}.\nonumber 
\end{align}
To see this, consider \Eqref{pStarAPprimeA}. Alice knows that if
she wins the first round, her probability of winning is $\alpha>\beta$.
She knows that in the first round, she can force the outcome $A$
with probability $\alpha$. From leniency, she knows that Bob would
output $B$ with the remaining probability.\footnote{Without leniency, this probability could have been shared between
the outcomes $\perp$ (abort) and $B$. Consequently, only upper bounds
could be obtained on $p_{A}^{*}(\mathcal{P}')$ and $p_{B}^{*}(\mathcal{P}')$
using only $\alpha$ and $\beta$ as security guarantees for $\mathcal{P}_{A}$.
Upper bounds, however, would not be enough to determine the polarity
of $\mathcal{P}'$ and stymie an unambiguous repetition of the composition
procedure (at least as it is defined). One could nevertheless compose
by hoping that the upper bounds can be used to accurately represent
the polarity. Regardless, this would still yield a protocol and the same calculation would yield correct bounds but the composition itself might be sub-optimal.}

A side remark: one consequence of this simplified analysis is that\footnote{$\alpha^{(1)}-\beta^{(1)}=(\alpha-\beta)\alpha-(\alpha-\beta)\beta=(\alpha-\beta)^{2}>0$}
$\alpha^{(1)}>\beta^{(1)}$. Intuitively, it means that polarity is
preserved by the composition procedure. Proceeding similarly, i.e.
defining $\mathcal{\mathcal{P}}'':=C^{LL}(\mathcal{P},\mathcal{P}')$,
and repeating $k+1$ times overall, one obtains\footnote{Again, note that $\alpha^{(k+1)}-\beta^{(k+1)}=(\alpha^{(k)}-\beta^{(k)})(\alpha-\beta)>0$.
} 
\begin{align*}
\alpha^{(k+1)} & =\alpha\alpha^{(k)}+(1-\alpha)\beta^{(k)}\\
\beta^{(k+1)} & =\beta\alpha^{(k)}+(1-\beta)\beta^{(k)}.
\end{align*}
In the limit of $k\to\infty$, one obtains 
\[
p_{A}^{*}(C^{LL}(\mathcal{P}))=p_{B}^{*}(C^{LL}(\mathcal{P}))=\lim_{k\to\infty}\alpha^{(k)}=\lim_{k\to\infty}\beta^{(k)}\approx0.8199.
\]
Proceeding similarly, one obtains for $X\in\{A,B\}$ and $\mathcal{X}\in\{\mathcal{W},\mathcal{Q}\}$,
\[
p_{X}^{*}(C^{LL}(\mathcal{X}))\approx0.836\,.
\]
We thus have the following. 
\begin{figure}
	
\begin{centering}
\includegraphics[scale = 1]{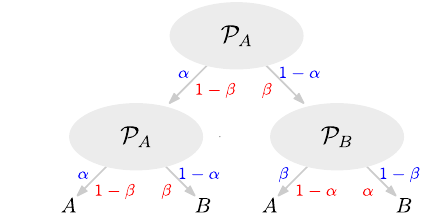}
\par\end{centering}
\caption{Standard composition of weak coin flipping protocols. Subprotocols only have two outcomes depending on the coin flip. Coloured labels indicate probabilities of outcomes for cheating Alice (blue) and cheating Bob (red).\label{fig:Standard-composition-technique}} 
\end{figure}
}
\begin{thm}\label{thm:CLL}
Let $\mathcal{W}$ and $\mathcal{Q}$ denote \Algref{WCF} and \Algref{BobSelfTests} respectively. Suppose $X\in\{A,B\}$ and $\mathcal{X}\in\{\mathcal{W},\mathcal{Q}\}$.
Then, asymptotically in the security parameter (i.e. in the limit that $\lambda\to \infty$), one has $p_{X}^{*}(C^{LL}(\mathcal{X})) \le 0.836$ and, assuming \Conjref{Qcont} holds, one has $p_{X}^{*}(C^{LL}(\mathcal{P})) \le 0.8199$. 
\end{thm}

{In \Thmref{CLL}, $\lambda$ is a security parameter that specifies the number of compositions, $k=\lambda$, and the number of devices used in the self-test step, $n=\lambda^3$. The theorem also uses a slightly different definition of $C^{LL}$ from the one described above. Details are deferred to \Secref{SecurityFiniteN} but the key idea remains the same.} 

\subsection{Abort Phobic Compositions | $C^{L\perp},C^{\perp L}$}

\branchcolor{black}{We now look at the case of abort phobic compositions, $C^{L\perp}$
and $C^{\perp L}$. We work through essentially the same example as
above and see what changes in this setting. As usual, let $\mathcal{P}$ denote \Algref{AliceSelfTests}
 and recall that as before 
\begin{align*}
p_{A}^{*}(\mathcal{P}_{A}) & =:\alpha\approx0.852,\\
p_{B}^{*}(\mathcal{P}_{A}) & =:\beta\approx0.667.
\end{align*}
In addition, we know from \Lemref{AliceSelfTests} that cheat vectors
for Bob, $(v_A,v_B,v_{\perp})\in\mathbb{C}_{B}(\mathcal{P}_{A})$
admit a nice characterisation courtesy of the self-testing step. Let
$\mathcal{P}':=C^{\perp L}(\mathcal{P},\mathcal{P})$, i.e. Alice
and Bob execute $\mathcal{P}_{A}$ and if the outcome is $A$, they
execute $\mathcal{P}_{A}$ while if the outcome is $B$, they execute
$\mathcal{P}_{B}$. Bob is assumed to be lenient so an honest Bob
never outputs abort, $\perp$. However, an honest Alice can output
abort, $\perp$ so we keep that output in the illustration, \Figref{Abort-Augmented-Composition}.
Our goal is to find $p_{A}^{*}(\mathcal{P}')$ and $p_{B}^{*}(\mathcal{P}')$.
The former is the same as before because Bob is lenient: 
\[
p_{A}^{*}(\mathcal{P}')=\alpha\cdot\alpha+(1-\alpha)\cdot\beta.
\]
Clearly, $p_{B}^{*}(\mathcal{P}')\le\beta\alpha+(1-\beta)\beta$ but
this bound may not be tight because $(1-\beta)$ is the combined probability
of Alice aborting and Alice outputting $A$. However, we can use cheat
vectors to obtain 
\[
p_{B}^{*}(\mathcal{P}')=\max_{(v_{A},v_{B},v_{\perp})\in\mathbb{C}_{B}}v_{B}\alpha+v_{A}\beta
\]
which is an SDP one can solve numerically. Unlike the previous case,
the polarity of the resulting protocol, $\mathcal{P}'$, might have
flipped (compared to the polarity of $\mathcal{P}$). 

Repeating this procedure, one can consider $\mathcal{P}'':=C^{\perp L}(\mathcal{P},\mathcal{P}')$
and obtain $p_{A}^{*}(\mathcal{P}'')$ directly as illustrated above
and numerically solve for $p_{B}^{*}(\mathcal{P}'')$ using the cheat
vectors. After around fifteen such iterations, we found that the cheating probabilities converged to approximately $0.81459$.
We also observed that the abort probabilities associated with $\mathcal{P}$
were quite small and therefore one could hope that $\mathcal{Q}$ (which denotes \Algref{BobSelfTests})
fares better. Proceeding analogously and considering $\mathcal{Q}':=C^{L\perp}(\mathcal{Q},\mathcal{Q})$,
$\mathcal{Q}'':=C^{L\perp}(\mathcal{Q},\mathcal{Q}')$, etc., we found that the
cheating probabilities converged to approximately $0.822655$. 

\begin{figure}
\begin{centering}
\includegraphics[scale=1]{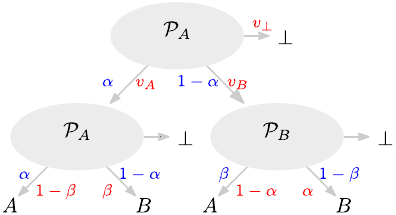}
\par\end{centering}
\caption{Abort phobic composition for weak coin flipping protocols. Subprotocols have three possible outcomes including abort. Aborting in any subprotocol directly leads to aborting the whole protocol. Coloured labels indicate probabilities of outcomes for cheating Alice (blue) and cheating Bob (red). In the security analysis of cheating Bob, we need to optimise over the cheat vectors  $(v_{A},v_{B},v_{\perp})\in\mathbb{C}_{B}$. 
  \label{fig:Abort-Augmented-Composition}}
\end{figure}
}

\begin{thm}\label{thm:Lperp_perpL}
Let $\mathcal{P}$ and $\mathcal{Q}$ denote \Algref{AliceSelfTests} and \Algref{BobSelfTests} respectively. Suppose $X\in\{A,B\}$. Then, {asymptotically in the security parameter (i.e. in the limit that $\lambda\to \infty$),} 
one has
\[
p_{X}^{*}(C^{\perp L}(\mathcal{P}))\le 0.81459
\]
assuming the continuity conjecture, \Conjref{Pcont} (for \Algref{AliceSelfTests}), holds. Further, 
\[
p_{X}^{*}(C^{L\perp}(\mathcal{Q}))\le 0.822655
\] 
 assuming the continuity conjecture, \Conjref{Qcont} (for \Algref{BobSelfTests}), holds.
\end{thm}

\branchcolor{black}{While by itself $\mathcal{Q}$ doesn't seem to help, one can suppress
the bias further, by noting that at the very last step, only the cheating
probabilities $p_{A}^{*}(\mathcal{Q})$ and $p_{B}^{*}(\mathcal{Q})$
played a role (i.e. the fact that the cheating vectors $\mathbb{C}_{A}$
for $\mathcal{Q}$ had an SDP characterisation was not used). Further,
we know that $p_{A}^{*}(\mathcal{P})=p_{A}^{*}(\mathcal{Q})$ but
$p_{B}^{*}(\mathcal{P})<p_{B}^{*}(\mathcal{Q})$, i.e. using $\mathcal{P}$
at the very last step will result in a strictly better protocol. }

\begin{thm}\label{thm:mixed_Lperp}
Let $\mathcal{P}$ and $\mathcal{Q}$ denote \Algref{AliceSelfTests} and \Algref{BobSelfTests} respectively. Suppose the continuity conjectures, \Conjref{Pcont,Qcont}, hold. Let $X\in\{A,B\}$, 
\begin{align*}
\mathcal{Z}^{1} & :=C^{L\perp}(\mathcal{Q},\mathcal{P}),\quad{\rm and}\\
\mathcal{Z}^{k+1} & :=C^{L\perp}(\mathcal{Q},\mathcal{Z}^{k})\quad i>1.
\end{align*} 
Then, {asymptotically in the security parameter (i.e. in the limit that $\lambda\to \infty$),} 
\[
\lim_{k\to\infty}p_{X}^{*}(\mathcal{Z}^{k}) \le 0.791044.
\]
\end{thm} 

{
As was the case for \Thmref{CLL}, in the theorems above, $\lambda$ controls the number of boxes, $n=\lambda^3$, used in the self-test step and the number of compositions $k=\lambda$. Furthermore, the theorems above also use slightly modified definitions for $C^{\perp L}$ and $C^{L \perp}$ as detailed in \Secref{SecurityFiniteN}. 
}

\section{Security proofs | Asymptotic limit}
\label{sec:securityAsymptotic}
The goal of this section is to analyse the security of \Algref{AliceSelfTests} (resp. \Algref{BobSelfTests}) in the simplest setting---assuming that Alice's box (resp. Bob's boxes) indeed correspond to the GHZ state and measurements. %
To be more precise, in this section, we assume \Assuref{GHZwinwithepsilon} with $\epsilon=0$.

\begin{assumption}[Restatement: \Assuref{GHZwinwithepsilon} with $\eps=0$]
\label{assu:asymptotic}In protocol $\mathcal{P}$ (resp. $\mathcal{Q}$),
Alice (resp. Bob) does not perform the box verification step and instead
it is assumed that her box is (resp. his boxes are) taken from a triple of
boxes which win the GHZ game with certainty. 
\end{assumption}

\branchcolor{black}{\Lemref{rigidityGHZ} asserts  that when the
winning probability is exactly one (i.e. $\epsilon=0$ in the lemma), the states and measurements are
the same as the GHZ state and $\sigma_{x},\sigma_{y}$ measurements,
up to local isometries and this allows us to use semi-definite programming. Below, we use the following notation.}
\begin{itemize}
	\item Quantum registers are denoted by capital letters, e.g. $A,B,C$.
	\item A pure quantum state in these registers is denoted as a vector $\ket{\psi} \in ABC$ where $ABC$ is interpreted to be the vector space corresponding to the registers $A,B,C$.
	\item A mixed quantum state is denoted by $\rho \in \Pos(ABC)$ where $\Pos(ABC)$ is the set of all Hermitian matrices on the vector space $ABC$ with non-negative eigenvalues.
\end{itemize}

\subsection{SDP when Alice self-tests | $\mathcal{P}$ \label{subsec:SDP-when-Alice}}

\branchcolor{black}{\begin{proof}[Asymptotic proof of \Lemref{AliceSelfTests}]
We prove \Lemref{AliceSelfTests} under \Assuref{asymptotic}. We
begin by making two observations. 

First, note that in the protocol, if Alice applies an isometry on
her box \emph{after} she has inputted $x$, obtained the outcome $a$
(and has noted it somewhere), the security of the resulting protocol
is unchanged because the rest of the protocol only depends on $x$
and $a$, and Alice's isometry only amounts to relabelling of the
post measurement state. This freedom allows us to simplify the analysis.

Second, in the analysis, we cannot model Alice's random choice, say
for $x$, as a mixed state because Bob can always hold a purification
and thus know $x$. Therefore, we model the randomness using pure
states and measure them in the end.

Notation: Other than $PQR$, all other registers store qubits.

We proceed step by step. 
\begin{enumerate}
\item We can model (justified below) Alice's act of inputting a random $x$
and obtaining an outcome $a$ from her box through the state 
\begin{equation}
\left| \Psi_{0} \right\rangle :=\frac{1}{2}\sum_{x,a\in\{0,1\}}\left|xa\right\rangle _{XA}\left|\Phi(x,a)\right\rangle _{IJ}
\end{equation}
where $X$ represents the random input and $A$ the output. Here,
$\left|\Phi(x,a)\right\rangle _{IJ}$ are Bell states (see \Eqref{bellStates})
and the registers $IJ$ are held by Bob. Alice's act of choosing $r$
at random, computing $s=a\oplus x.r$ is modelled as 
\begin{equation}
\left|\Psi_{1}\right\rangle :=\frac{1}{2\sqrt{2}}\sum_{x,a,r\in\{0,1\}}\left|xa\right\rangle _{XA}\left|\Phi(x,a)\right\rangle _{IJ}\left|r\right\rangle _{R}\left|a\oplus x.r\right\rangle _{S}.\label{eq:Alice_Psi1}
\end{equation}
Finally, Alice's act of sending $s$ is modelled as Alice starting
with the state
\[
\tr_{IJS}\left[\left|\Psi_{1}\right\rangle \left\langle \Psi_{1}\right|\right]\in \Pos(XAR).
\] 

\textbf{Justification for starting with $\left|\Psi_{0}\right\rangle $.}\\
To see why we start with the state $\left|\Psi_{0}\right\rangle $,
model Alice's choice of $x$ as $\left|+\right\rangle _{X}$, suppose
her measurement result is stored in $\left|0\right\rangle _{A}$,
the state of the boxes before measurement is $\left|\psi\right\rangle _{PQR}$
and Alice holds $P$, i.e. 
\[
\left|\Psi_{0}'\right\rangle :=\left|+\right\rangle _{X}\left|0\right\rangle _{A}\left|\psi\right\rangle _{PQR}.
\]
Let $\{M_{a|x}^{P}\}$ be the measurement operators corresponding
to Alice's box. The measurement process is unitarily modelled as 
\[
\left|\Psi_{1}'\right\rangle :=U_{{\rm measure}}\left|\Psi_{0}'\right\rangle =\frac{1}{\sqrt{2}}\sum_{x,a\in\{0,1\}}\left|x\right\rangle _{X}\left|a\right\rangle _{A}M_{a|x}^{P}\left|\psi\right\rangle _{PQR}
\]
where
\[
U_{{\rm measure}}=\sum_{x\in\{0,1\}}\left|x\right\rangle \left\langle x\right|_{X}\otimes\left[\mathbb{I}_{A}\otimes M_{0|x}^{P}+X_{X}\otimes M_{1|x}^{P}\right]\otimes\mathbb{I}_{QR}.
\]
Now we harness the freedom of applying an isometry to the post measured
state (as observed above). We choose the local isometry in \Lemref{rigidityGHZ}.
Without loss of generality, we can assume that Bob had already applied
his part of the isometry before sending the boxes (because he can
always reverse it when it is his turn). We thus have, 
\begin{align*}
\left|\Psi_{2}'\right\rangle :=\Phi_{PQR}\left|\Psi_{1}'\right\rangle  & =\frac{1}{\sqrt{2}}\sum_{x,a\in\{0,1\}}\left|x\right\rangle _{X}\left|a\right\rangle _{A}\Pi_{x|a}^{H}\left|{\rm GHZ}\right\rangle _{HIJ}\otimes\left|{\rm junk}\right\rangle _{PQR}\\
 & =\frac{1}{2}\sum_{x,a\in\{0,1\}}\left|x\right\rangle _{X}\left|a\right\rangle _{A}U^{H}(x,a)\left|0\right\rangle _{H}\left|\Phi(x,a)\right\rangle _{IJ}\otimes\left|{\rm junk}\right\rangle _{PQR}
\end{align*}
where 
\begin{equation}
\left|\Phi(x,a)\right\rangle _{IJ}=\frac{\left|00\right\rangle +(-1)^{a}(i)^{x}\left|11\right\rangle }{\sqrt{2}}\label{eq:bellStates}
\end{equation}
 and $U^{H}(x,a)\left|0\right\rangle _{H}$ is $\frac{\left|0\right\rangle +(-1)^{a}(i)^{x}\left|1\right\rangle }{\sqrt{2}}$.
Since the state of register $H$ is completely determined by registers
$X$ and $A$, we can drop it from the analysis without loss of generality.
Finally, since $\left|{\rm junk}\right\rangle _{PQR}$ is completely
tensored out, we can drop it too without affecting the security. Formally,
we can assume that Alice gives Bob the register $P$ at this point. 

\item Bob sending $g$ is modelled by introducing $\rho_{2}\in \Pos(XARG)$ satisfying
$\tr_{IJS}\left[\left|\Psi_{1}\right\rangle \left\langle \Psi_{1}\right|\right]=\tr_{G}(\rho_{2})$. 

\item At this point, either $x\oplus g$ is zero, in which case Alice's
output is fixed or $x\oplus g$ is one, and in that case Bob will
already know $x$ because he knows $g$ (he sent it) and Alice will
proceed to testing Bob. Formally, therefore, we needn't do anything
at this step.

\item Assuming $x\oplus g=1$, Alice sends $y,z$ to Bob such that $x\oplus y\oplus z=1$.
However, since Bob already knows $x$, he can deduce $z$ from $y$.
We thus only need to model Alice sending $y$ and Bob responding with
$d=b\oplus c$ (because Alice will only use $b\oplus c$ to test the
GHZ game, so it suffices for Bob to send $d$). This amounts to introducing
$\rho_{3}\in \Pos(XARGYD)$ satisfying $\rho_{2}\otimes\frac{\mathbb{I}_{Y}}{2}=\tr_{D}(\rho_{3})$.

\item Since we postponed the measurements to the end, we add this last step.
Alice now measures $\rho_{3}$ to determine $x\oplus g$ and if it
is one, whether the GHZ test passed. Let 
\begin{align}
    \Pi_{i} & :=\sum_{x,y\in\{0,1\}:x\oplus g=i}\left|xg\right\rangle \left\langle xg\right|_{XG}\otimes\mathbb{I}_{ARYD}, \\
    \Pi^{{\rm GHZ}} & :=\sum_{\substack{
            x,y\in\{0,1\},\\a,d\in\{0,1\}:a\oplus d\oplus1=xy\cdot(1\oplus x\oplus y)}
        }
        \left|xyad\right\rangle \left\langle xyad\right|_{XYAD}\otimes\mathbb{I}_{RG}.
    \label{eq:AliceProjs}
\end{align}
Then, we can write the cheat vector for Alice, i.e. the tuple of probabilities
that Alice outputs 0, 1 and abort (see \Defref{CheatVectors}), as
\[
    (\alpha,\beta,\gamma)=(\tr(\Pi_{0}\rho_{3}),\tr(\Pi_{1}\Pi^{{\rm GHZ}}\rho_{3}),\tr(\Pi_{1}\bar{\Pi}^{{\rm GHZ}}\rho_{3}))
\]
 where $\bar{\Pi}:=\mathbb{I}-\Pi$.
\end{enumerate}
To summarise, the final SDP is as follows: let $\left|\Psi_{1}\right\rangle \in \Pos(XAIJRS)$
be as given in \Eqref{Alice_Psi1}, $\rho_{2}\in \Pos(XARG)$ and $\rho_{3}\in \Pos(XARGYD)$
\begin{equation}
    \eta:= \, \max \, \tr([c_{0}\Pi_{0}+\Pi_{1}(c_{1}\Pi^{{\rm GHZ}}+c_{\perp}\bar{\Pi}^{{\rm GHZ}})]\rho_{3}) \label{eq:SDP_AliceSelfTests}
\end{equation}
subject to 
\begin{align*}
    \tr_{IJS}\left[\left|\Psi_{1}\right\rangle \left\langle \Psi_{1}\right|\right] & =\tr_{G}(\rho_{2})\\
    \rho_{2}\otimes\frac{\mathbb{I}_{Y}}{2} & =\tr_{D}(\rho_{3})
\end{align*}
where the projectors are defined in \Eqref{AliceProjs}.
\end{proof}
}

\subsection{SDP when Bob self-tests | $\mathcal{Q}$ \label{subsec:SDP-when-Bob}}
 
\begin{proof}[Asymptotic proof of \Lemref{BobSelfTests}]
Denote by $\mathcal{S}$ the protocol corresponding to \Algref{SCF}. 

It is evident that $p_{B}^{*}(\mathcal{Q})\le p_{B}^{*}(\mathcal{S})$
because compared to $\mathcal{S}$, in $\mathcal{Q}$ Alice performs
an extra test. However, it is not hard to construct a cheating strategy for Bob which lets him saturate the inequality, i.e. $p_{B}^{*}(\mathcal{Q})=p_{B}^{*}(\mathcal{S})$. 

From \Lemref{SCFstandard}, it is also clear that $p_{A}^{*}(\mathcal{Q})=p_{A}^{*}(\mathcal{S})$
because the only difference between Bob's actions in $\mathcal{Q}$
and $\mathcal{S}$ is that Bob self-tests to ensure his boxes are
indeed GHZ. However, the optimal cheating strategy for $\mathcal{S}$
can be implemented using GHZ boxes. 

This establishes the first part of the lemma. For the second part,
i.e. establishing that optimising $c_{0}\alpha+c_{1}\beta+c_{\perp}\gamma$
over $(\alpha,\beta,\gamma)\in\mathbb{C}_{A}$ is an SDP, we proceed
as follows. Suppose \Assuref{asymptotic} holds. Then we can assume
that Bob starts with the state 
\begin{equation}
\rho_{0}:=\tr_{H}(\left|{\rm GHZ}\right\rangle \left\langle {\rm GHZ}\right|_{HIJ})\label{eq:Bob_initState}
\end{equation}
 and the effect of measuring the two boxes can be represented by the
application of projectors of pauli operators $X$ and $Z$.

The justification is similar to that given in the former proof. Suppose
Bob holds registers $QR$ of $\left|\psi\right\rangle _{PQR}$ which
is the combined state of the three boxes. Suppose his measurement
operators are $\{M_{b|y}^{Q},M_{c|z}^{R}\}$. Then using the isometry
in \Lemref{rigidityGHZ}, Bob can relabel his state (and without
loss of generality, we can suppose Alice also relabels according to
the aforementioned isometry) to get $\Phi_{PQR}\left|\psi\right\rangle _{PQR}=\left|{\rm GHZ}\right\rangle _{HIJ}\otimes\left|{\rm junk}\right\rangle _{PQR}$.
Further, since $\Phi_{PQR}(M_{b|y}^{Q}\otimes M_{c|z}^{R}\left|\psi\right\rangle _{PQR})=\Pi_{b|y}^{I}\Pi_{c|z}^{J}\left|{\rm GHZ}\right\rangle _{HIJ}\otimes\left|{\rm junk}\right\rangle _{PQR}$
Bob's act of measurement, in the new labelling, corresponds to simply
measuring the GHZ state in the appropriate Pauli basis. 
\begin{enumerate}
\item Bob receiving $s$ from Alice is modelled by introducing $\rho_{1}\in \Pos(SIJ)$
satisfying $\tr_{S}(\rho_{1})=\rho_{0}$. 
\item Bob sending $g\in_{R}\{0,1\}$ can be seen as appending a mixed state:
$\rho_{1}\otimes\frac{1}{2}\mathbb{I}_{G}$.
\item Alice sending $x$ (and $a$) can be modelled as introducing $\rho_{2}\in \Pos(AXSIJG)$
satisfying $\tr_{A}(\rho_{2})=\rho_{1}\otimes\frac{\mathbb{I}_{G}}{2}$.
\item To model the GHZ test, introduce a register $Y$ in the state $\frac{\left|0\right\rangle _{Y}+\left|1\right\rangle _{Y}}{\sqrt{2}}$.
Recall that to perform the GHZ test, we need $x\oplus y\oplus z=1$
i.e. $z=1\oplus y\oplus x$. Further introduce registers $B$ and
$C$ to hold the measurement results, define 
\begin{equation}
U:=\sum_{y,x\in\{0,1\}}\left|y\right\rangle \left\langle y\right|_{Y}\left|x\right\rangle \left\langle x\right|_{X}\otimes(\mathbb{I}_{B}\otimes\Pi_{0|y}^{I}+X_{B}\otimes\Pi_{1|y}^{I})\otimes(\mathbb{I}_{C}\otimes\Pi_{0|(1\oplus y\oplus x)}^{J}+X_{C}\otimes\Pi_{1|(1\oplus y\oplus x)}^{J})\otimes\mathbb{I}_{ASG}.\label{eq:Bob-u}
\end{equation}
By construction, $\rho_{3}:=U\left(\left|+\right\rangle \left\langle +\right|_{Y}\otimes\left|00\right\rangle \left\langle 00\right|_{BC}\otimes\rho_{2}\right)U^{\dagger}\in \Pos(YBCAXSIJG)$
models the measurement process. 
\item Since we postponed the measurements to the end, we add this step.
Define 
\[
\Pi_{i}:=\sum_{x,g\in\{0,1\}:x\oplus g=i}\left|xg\right\rangle \left\langle xg\right|_{XG}\otimes\mathbb{I}_{YABSIJ}
\]
to determine who won. Define 
\[
\Pi^{{\rm sTest}}:=\sum_{s,a,x\in\{0,1\}:s=a\lor s=a\oplus x}\left|sax\right\rangle \left\langle sax\right|_{SAX}\otimes\mathbb{I}_{GYBCIJ}
\]
 to model the first test, i.e. $s$ should either be $a$ or $a\oplus x$.
Define 
\[
\Pi^{{\rm GHZ}}:=\sum_{\substack{x,y\in\{0,1\},\\
a,b,c\in\{0,1\}:a\oplus b\oplus c\oplus1=xy\cdot(1\oplus x\oplus y)
}
}\left|xyabc\right\rangle \left\langle xyabc\right|_{XYABC}\otimes\mathbb{I}_{GSIJ}
\]
to model the GHZ test. Let 
\begin{equation}
\Pi^{{\rm Test}}:=\Pi^{{\rm GHZ}}\Pi^{{\rm sTest}},\quad\bar{\Pi}^{{\rm Test}}:=\mathbb{I}-\Pi^{{\rm Test}}.\label{eq:BobProjs}
\end{equation}
One can then write the cheat vector for Bob, i.e. the tuple of probabilities
that Bob outputs $0,1$ and abort (see \Defref{CheatVectors}), as
\[
(\alpha,\beta,\gamma)=(\tr(\Pi_{0}\Pi^{{\rm Test}}\rho_{3}),\tr(\Pi_{1}\rho_{3}),\tr(\Pi_{0}\bar{\Pi}^{{\rm Test}}\rho_{3})).
\]
\end{enumerate}
 
To summarise, the final SDP is as follows: let $\rho_{0}\in \Pos(IJ)$ be
as defined in \Eqref{Bob_initState}, $\rho_{1}\in \Pos(SIJ)$ and $\rho_{2}\in \Pos(AXSIJG)$.
Then, 
\begin{equation}
\eta := \max\quad\tr\left([\Pi_{0}(c_{0}\Pi^{{\rm Test}}+c_{\perp}\bar{\Pi}^{{\rm Test}})+c_{1}\Pi_{1}]U\left(\left|+00\right\rangle \left\langle +00\right|_{YBC}\otimes\rho_{2}\right)U^{\dagger}\right)
\label{eq:etaBobSelfTests} \end{equation}
subject to 
\begin{align*}
\tr_{S}(\rho_{1}) & =\rho_{0}\\
\tr_{A}(\rho_{2}) & =\frac{1}{2}\rho_{1}\otimes\mathbb{I}_{G}
\end{align*}
where $U$ is as defined in \Eqref{Bob-u} and the projectors as in
\Eqref{BobProjs}.
\end{proof}

\section{Security Proofs | Finite $\lambda$} 
\label{sec:SecurityFiniteN}

{
In this section, we complete the missing steps from the asymptotic analysis. 
\begin{itemize}
    \item We already derived the cheat vector SDP characterisation for \Algref{AliceSelfTests} in \Lemref{AliceSelfTests} (resp. for \Algref{BobSelfTests} in \Lemref{BobSelfTests}) assuming the boxes win GHZ with certainty, i.e. \Assuref{asymptotic}, instead of assuming the boxes win with probability $1-\epsilon$, i.e. \Assuref{GHZwinwithepsilon}. \Subsecref{SDP-when-Alice-N} (resp. \Subsecref{SDP-when-Bob-N}) precisely states the conjecture that using \Assuref{GHZwinwithepsilon}, the security analysis of \Algref{AliceSelfTests} (resp. \Algref{BobSelfTests}) converges to that done using \Assuref{asymptotic}.
    \item \Subsecref{EstimateGHZ}, shows how to use the step where $n-1$ devices are tested, to draw the conclusion that the remaining device wins the GHZ game with probability $1-\epsilon$, justifying \Assuref{GHZwinwithepsilon}. The analysis here holds quite generally for any game with perfect completeness and may be of independent interest. 
    \item Finally, \Subsecref{assembling}, uses this link, \Lemref{AliceSelfTests} (resp. \Lemref{BobSelfTests}), and \Conjref{Pcont} (resp. \Conjref{Qcont}) to prove the security of \Algref{AliceSelfTests} as stated in \Propref{AliceSelfTests}. %
	It ends by completing the remaining steps in \Secref{Second-Technique} to establish the proofs of security for the composed protocols, as stated in \Thmref{CLL, Lperp_perpL, mixed_Lperp}.
\end{itemize}
} 
 
\subsection{Continuity conjecture when Alice self-tests | $\mathcal{P}$}
\label{subsec:SDP-when-Alice-N}

\begin{conjecture}[Continuity conjecture for $\mathcal{P}$]\label{conj:Pcont} Let $c_0,c_1,c_{\perp}$ be non-negative. Denote by $v_0,v_1,v_{\perp}$ the probability that Alice outputs $0,1,\perp$ respectively when \Algref{AliceSelfTests} is executed against some cheating strategy of Bob. Let $0\le 1-\eps'(N) \le 1$ denote the winning probability of the GHZ boxes held by Alice {(via \Assuref{GHZwinwithepsilon}).} 
	Let $\epsilon := f(\epsilon')$ where $f$ is as in \Lemref{rigidityGHZ}. {Consider the objective $\eta_{\epsilon}=\max c_0v_0 + c_1v_1 + c_{\perp}v_{\perp}$ where the maximisation is over $v_0,v_1,v_\perp$. Then, $\lim_{\epsilon \to 0} \eta_{\eps} = \eta$ where $\eta$ is the value of the SDP in \Eqref{SDP_AliceSelfTests}.} 
\end{conjecture}

{\Conjref{Pcont} and the asymptotic analysis in \Subsecref{SDP-when-Alice} immediately yield \Lemref{AliceSelfTests}.}

\subsection{Continuity conjecture when Bob self-tests | $\mathcal{Q}$\label{subsec:SDP-when-Bob-N}}

The analogous statement for the case where Bob self-tests is the following.

\begin{conjecture}[Continuity conjecture for $\mathcal{Q}$]\label{conj:Qcont} Let $c_0,c_1,c_{\perp}$ be non-negative. Denote by $v_0,v_1,v_{\perp}$ the probability that Bob outputs $0,1,\perp$ respectively when \Algref{BobSelfTests} is executed against some cheating strategy of Alice. Let $0\le 1-\eps' \le 1$ denote the winning probability of the GHZ boxes held by Bob {(via \Assuref{GHZwinwithepsilon}).} 
Let $\epsilon := f(\epsilon')$ where $f$ is as in \Lemref{rigidityGHZ}. {Consider the objective $\eta_{\epsilon}=\max c_0v_0 + c_1v_1 + c_{\perp}v_{\perp}$ where the maximisation is over $v_0,v_1,v_\perp$. Then, $\lim_{\epsilon \to 0} \eta_{\eps} = \eta$ where $\eta$ is the value of the SDP in \Eqref{etaBobSelfTests}.} 
\end{conjecture}

{\Conjref{Qcont} and the asymptotic analysis in \Subsecref{SDP-when-Bob} immediately yield \Lemref{BobSelfTests}.}

\subsection{Estimation of GHZ winning probability}
\label{subsec:EstimateGHZ}

We restate the self-testing procedure using notation tailored to the analysis here. We assume that the $3n$ boxes are described by some joint quantum state and local measurement operators. After playing the GHZ game with $3(n-1)$ of them, and verifying that they all pass, we want to make a statement about the remaining box, whose state $\tilde \rho$ is conditioned on the passing of all the other test. 

\begin{varalgorithm}{Est-GHZ} 
 \caption{\label{alg:self-test}}
 	\begin{enumerate}
		\item Pick a box $J \in [ n ]$ uniformly at random.
		\item For $i \in [n]\backslash J$, play the GHZ game with box $i$, denote outcome of game as $X_i\in \{0,1\}$ 
		\item If 
		\begin{IEEEeqnarray}{r'L}\label{eq:Omega}
		     \Omega:& X_i = 1 \text{, for all } i\in [n] \backslash J
		\end{IEEEeqnarray}
		\item Then conclude that the remaining box satisfies
		\begin{IEEEeqnarray}{r'L}\label{eq:T}
		     T:& E[X_J|J,\Omega] \geq 1 - \epsilon 
		\end{IEEEeqnarray}
	\end{enumerate}
\end{varalgorithm}
	
We remark that the expectation value of $E[X_J|J,\Omega]$ accurately describes the expected GHZ value associated to the state of the remaining boxes $J$, conditioned on having measuring some outcome sequence in the other boxes which passes all the GHZ tests. Note that the conditioning in $J$ is important because otherwise we would get a bound on the GHZ averaged over all boxes, but we are only interested in the remaining box.

\begin{prop}[Security of \Algref{self-test}]
	\label{prop:security}
	For any implementation of the boxes and choice of $\epsilon>0$ the joint probability that that the test $\Omega$ passes and that the conclusion $T$ is false is small, more precisely, $\Pr[ \Omega \cap \bar T] \leq \frac{1}{1-\epsilon + n\epsilon} \leq \frac{1}{n\epsilon}$, where the first upper-bound is tight.
\end{prop}

This is the correct form of the security statement. It is important to bound the joint distribution of $\Omega$ and $\overline T$, and not $\Pr[\overline T|\Omega]$, conditioning on passing the test $\Omega$. Indeed in the latter case, it would not be possible to conclude anything of value about the remaining box $J$, as there could be some implementation of the boxes which has a very low expectation value of GHZ, but which passes the test with small but non-zero probability. 
Consider a hypothetical ideal protocol, which after having chosen $J$, only passes when $T$ is true. In that case, $\Pr[\Omega\cap \overline T] = 0$. Then the actual protocol is equivalent the ideal one, except that it fails with probability $\delta = \frac{1}{1-\epsilon + m\epsilon}$, and so it is $\delta$-close to the ideal algorithm. 

\begin{proof}
	For a given implementation of the boxes, let $p(x_1,\cdots x_n)$ denote the joint probability distribution of passing the GHZ games. Let $S = \{j|\ E[X_j|J=j, \Omega] < 1-\epsilon\} \subset [n]$ be the set of boxes that have an expectation value for GHZ (conditioned on passing in the other boxes) below our target threshold and let $m = |S|$ be the number of such boxes. The value of $m$ is unknown, so we will need to maximise over it in the end.
	
	Let $\alpha = \Pr(\{X_i\}_i = 1)$ and $\beta_j = \Pr(\{X_i\}_{i\neq j} = 1 \cap X_j = 0)$ be respectively the probabilities of the events where all the tests pass, or they all pass except for the $j$th test. This allows us to rewrite $E[X_j|J=j,\Omega] = \Pr(\{X_i\}_i=1)/\Pr(\{X_i\}_{i\neq j} = 1) = \alpha/(\alpha + \beta_j)$, and so, by definition of $S$, we have $ \alpha/(\alpha + \beta_j)<(1-\epsilon)$, for $j\in S$, which is equivalent to $\beta_j > \frac{\epsilon}{1-\epsilon} \alpha$. 
	
	The aim of the proof is to bound the probability $\Pr[\Omega \cap \overline T]$. If we condition and summed over the different values of $J$, we can rewrite it as
    \begin{IEEEeqnarray}{rL}
         \Pr(\Omega \cap \overline T) = \sum_j \frac{1}{n} \Pr(\Omega \cap \overline{T}| J = j) = \sum_{j\in S} \frac{1}{n} \Pr(\{X_i\}_{i\neq j} = 1) = \frac{1}{n} \sum_{j\in S} (\alpha + \beta_i)\,,
    \end{IEEEeqnarray}
    where we have kept the round $j\in S$ ones, conditioned on which $T$ is false. 
    We are thus left with the optimisation problem 
	\begin{IEEEeqnarray}{L'L}
		\max_{\alpha\geq 0,(\beta_i)_i\geq 0} 	
		    &\frac{1}{n}\left( \sum_{j\in S} \alpha + \beta_j\right)\\
		\mathrm{subject\ to}					
		    &\alpha + \sum_{j\in S} \beta_j \leq 1\\
			& \beta_j \geq \frac{\epsilon}{1-\epsilon} \alpha \text{, for }j\in S		
	\end{IEEEeqnarray}	    
	This is a linear problem. Simplifying it by defining $\Sigma = \sum_{j\in S} \beta_j$, gives
	\begin{IEEEeqnarray}{L'L}
		\max_{\alpha\geq0,\Sigma\geq0} 	&\frac{1}{n}( m\alpha + \Sigma)\\
		\mathrm{subject\ to}					
		    &\alpha + \Sigma \leq 1\\
			& \Sigma \geq m\frac{\epsilon}{1-\epsilon} \alpha			
	\end{IEEEeqnarray}
	It is easily shown that the maximum is attained for $(\alpha,\Sigma) = \left(\frac{1-\epsilon}{1-\epsilon + m\epsilon}, \frac{m \epsilon}{1-\epsilon + m\epsilon}\right)$ which gives the upper-bound
	\begin{IEEEeqnarray}{rL}
		\Pr[\Omega\cap \overline T] \leq \frac{1}{n} \max_m \frac{m}{1-\epsilon + m\epsilon} = \frac{1}{1-\epsilon + n\epsilon}
	\end{IEEEeqnarray}
	We note that the upper-bound is an increasing function of $m$ and so the maximum is attained for $m=n$. This yield the desired upper-bound. From the converse statement, we note that from the present proof we can construct a probability distribution $p(x_1,\cdots x_n)$, which saturates all inequalities, and so the upper-bound $\frac{1}{1-\epsilon + n\epsilon}$ is tight.
\end{proof}

\subsection{Putting everything together} \label{subsec:assembling}

{Using \Lemref{AliceSelfTests} (analysis of \Algref{AliceSelfTests} under \Assuref{GHZwinwithepsilon}) and \Propref{security} (in \Subsecref{EstimateGHZ} above), we show how to prove the security of \Algref{AliceSelfTests}, i.e. \Propref{AliceSelfTests}. } %

\vspace{1em}

{
Notation:
\begin{itemize}
    \item From the previous subsection, recall the definition of event $\Omega$ (see \Eqref{Omega}) which indicated whether the self-test step cleared, and that of $T$ (see \Eqref{T}) which indicated whether the last device (hypothetically) clears the GHZ test with probability at least $1-\epsilon$. 
    \item Denote by $\epsst:=\Pr(\Omega \bar T)$ (which, using \Propref{security} we will bound by $1/n\epsilon$).
    \item Denote by $p^{*\epsilon}_{A}$ (resp. $p^{*\epsilon}_B$) the probability that Alice (resp. Bob) wins when \Algref{AliceSelfTests} is executed under \Assuref{GHZwinwithepsilon}. 
\end{itemize}
}

\begin{proof}[Proof of \Propref{AliceSelfTests}]

    {
    We start by considering \Algref{AliceSelfTests} and label it as in \Figref{AliceSelfTests_}. It consists of two parts $\Pst$ and $\Peps$. 
    \begin{itemize}
        \item $\Pst$ outputs either ``abort'' or ``continue'' which are denoted by $\bar{\Omega}$ and $\Omega$ respectively.
        \item $\Peps$ is only executed on event $\Omega$ and it outputs either $A, B$ or $\perp$. 
    \end{itemize} 
    Note that, $\lim_{\epsilon\to0} p^{*\epsilon}_B$ can be calculated using \Lemref{AliceSelfTests} and $p^{*\epsilon}_A=p^*_A(\cal{W})$ since the behaviour of honest Bob is identical for the two protocols. Note also that, essentially by definition, $p^*_{A/B}(\Peps|T) = p^{*\epsilon}_{A/B}$. }

    {
    Now, $p^*_A(\calP) \le p^{*\epsilon}_A$ because Bob is honest and so $T$ is always true. Further, a cheating Alice can only decrease her probability of success by having Bob output $\bar \Omega$.}

    {
    As for $p^*_B(\calP)$, one can write 
    \begin{align*}
        p^*_B(\calP) &= p^*_B(\calP | \Omega) \Pr(\Omega) + p^*_B(\calP | \bar \Omega) \Pr(\bar \Omega) \\
        &\le p^*_B(\calP | \Omega T) + \Pr(\Omega \bar T) \\
        &\le p^{*\eps}_B + \eps^{\rm st}
    \end{align*}
    where in the first equality, note that $p^*_B(\calP | \bar \Omega)=0$ and from \Propref{security}, $\eps^{\rm st} = 1/(n\epsilon)$ bounds the probability that the self-test passed but the conclusion $T$ the device was wrong. Choosing $\epsilon=1/\lambda$ and $n=\lambda^3$, and taking the limit $\lambda \to \infty$, we obtain the bound asserted in \Propref{AliceSelfTests}. 
    }
    
    \begin{figure}[H]
        \centering
        \includegraphics[scale=1.0]{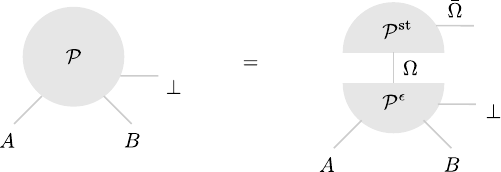}
        \caption{\Algref{AliceSelfTests} broken into two phases: the self-testing phase and the remaining protocol. Above, $\Omega$ denotes the outcome that the self-test phase succeeded and $\bar \Omega$ that it failed. $\Peps$ denotes continuation of the remaining protocol to output $A,B$ or $\perp$ (abort).} 
        \label{fig:AliceSelfTests_}
    \end{figure}

\end{proof}

{
We now show how to compose protocols involving the self-test step and conclude the proofs of \Thmref{CLL, Lperp_perpL, mixed_Lperp}. }

\begin{proof}[Proof of \Thmref{CLL, Lperp_perpL, mixed_Lperp}]
{
We start by revisiting lenient compositions. Suppose we are given a protocol $\calZ$ with $p^*_A(\calZ)=\alpha$ and $p^*_B(\calZ)=\beta$, with $\alpha>\beta$ as in \Figref{protocolZ}. Define the lenient composition of $\calZ$ with $\calP$, $C^{LL}(\calP,\calZ)$ as shown in \Figref{LenientComposition}. We have 
\begin{align*}
    p^*_A(C^{LL}(\calP,\calZ) &\le p^{*\eps}_A \cdot \alpha + (1-p^{*\eps}_A) \cdot \beta \\
    p^*_B(C^{LL}(\calP,\calZ) &\le (p^{*\eps}_B\cdot \alpha + (1-p^{*\eps}_B)\cdot \beta) \cdot \Pr(\Omega T) + \Pr(\Omega \bar T) \\
    & \le p^{*\eps}_B\cdot \alpha + (1-p^{*\eps}_B)\cdot \beta + \eps^{\text{st}}
\end{align*}
where the analysis for malicious Alice is the same as before and the analysis for malicious Bob follows by simply writing out the two possibilities that could lead Alice to output $B$, to wit: event $\Omega$ which can be split into two conditioned on $T$ as $\Omega T$ and $\Omega \bar T$. In both cases, one can bound the probability that Alice outputs $B$ by using \Lemref{AliceSelfTests} and \Propref{security}. 
}

{
One can similarly analyse the abort phobic composition $C^{\perp L}(\calP,\calZ)$ as shown in \Figref{AbortPhobicComposition}. Reasoning as above, one can write 
\begin{align*}
    p^*_A(C^{\perp L}(\calP,\calZ)) &\le p^{*\eps}_A \alpha + (1 - p^{*\eps}_A) \beta \\
    p^*_B(C^{\perp L}(\calP,\calZ)) &\le \max (v_A \cdot \beta + v_B \cdot \alpha) + \eps^{\rm st}
\end{align*}
where the maximisation is over cheat vectors of $\Peps$, i.e. $(v_A,v_B,v_{\perp})\in\mathbb{C}_B(\Peps)$, and we again used \Lemref{AliceSelfTests} and \Propref{security}.
}

{
Clearly, iterating this procedure, in both cases, results in an additional $\eps^{\rm st}$ term each time to the final cheating probabilities. Thus, for $k$ compositions, one must add $k \eps^{\rm st}$. For $k=\lambda$, $\epsilon=1/\lambda$ and $n=\lambda^3$, one obtains $\lim_{\lambda \to \infty} k \eps^{\rm st} = 0$ and the compositions converge to the asymptotic case that we already analysed in \Secref{Second-Technique}. This completes the proof of \Thmref{CLL, Lperp_perpL, mixed_Lperp}.
}

    \begin{figure}[H]
        \centering
        \includegraphics[scale=1.0]{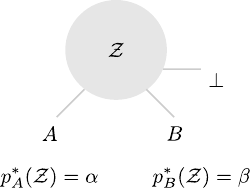}
        \caption{Suppose for protocol $\calZ$, one knows $p^*_A(\calZ)=\alpha$ and $p^*_B(\calZ)=\beta$ where $\alpha>\beta$.} 
        \label{fig:protocolZ}
    \end{figure}

    \begin{figure}[H]
        \centering
        \includegraphics[scale=1.0]{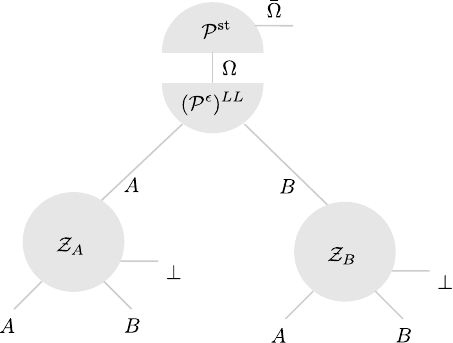}
        \caption{The lenient composition $C^{LL}(\calP,\calZ)$ is shown above. Even though the output $\bar \Omega$ is still an abort, the composition is lenient because $(\Peps)^{LL}$ means that in $\Peps$, the honest player replaces $\perp$ by declaring themselves the winner.} 
        \label{fig:LenientComposition}
    \end{figure}

    \begin{figure}[H]
        \centering
        \includegraphics[scale=1.0]{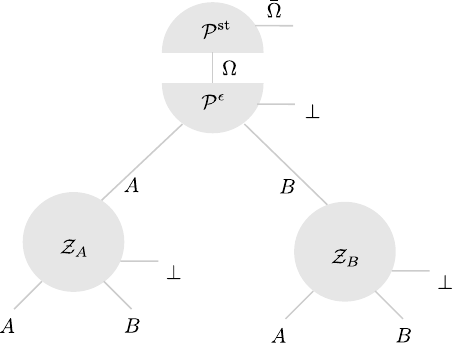}
        \caption{The abort-phobic composition $C^{\perp L}(\calP,\calZ)$ is shown above where $\Peps$ can output abort.} 
        \label{fig:AbortPhobicComposition}
    \end{figure}

\end{proof}

\section{Acknowledgements}
ASA is grateful for the hospitality of Perimeter Institute where part of this work was carried out. Research at Perimeter Institute is supported in part by the Government of Canada through the Department of Innovation, Science and
Economic Development Canada and by the Province of Ontario through the Ministry of Economic
Development, Job Creation and Trade. ASA further acknowledges the FNRS (Belgium) for support through the FRIA grants, 3/5/5 --
MCF/XH/FC -- 16754 and F 3/5/5 -- FRIA/FC -- 6700 FC 20759. A part of this work was carried out while ASA was at the Universit\'{e} libre de Bruxelles, Belgium. 
ASA acknowledges support from IQIM, an NSF Physics Frontier Center (GBMF-1250002), MURI grant FA9550-18-1-0161 and the U.S. Department of Defense through a QuICS Hartree Fellowship. 
Part of the work was carried out while ASA was visiting the Simons Institute for the Theory of Computing. TVH acknowledges support from ParisRegionQCI, supported by Paris Region; 
FranceQCI, supported by the European Commission, under the Digital Europe program; QSNP: 
European Union's Horizon Europe research and innovation program under the project ``Quantum Security Networks Partnership''.
JS is partially supported by Commonwealth Cyber Initiative SWVA grant 467489.

\pagebreak

\bibliographystyle{amsalpha}
\bibliography{DI_WCF_ideas}

\pagebreak

\appendix

\section{Device independence and the box paradigm}
\label{sec:BoxParadigm}
\branchcolor{black}{We describe device independent protocols as classical protocols with one modification: we assume that the two parties can exchange
boxes and that the parties can shield their boxes (from the other
boxes i.e. the boxes can't communicate with each other once shielded). 
}

\begin{defn}[Box]
 \label{def:box}A \emph{box} is a device that takes an input $x\in\mathcal{X}$
and yields an outputs $a\in\mathcal{A}$ where $\mathcal{X}$ and
$\mathcal{A}$ are finite sets. Typically, a set of $n$ boxes, taking
inputs $x_{1},x_{2},\dots x_{n}$ and yielding outputs $a_{1},a_{2}\dots a_{n}$
are \emph{characterised} by a joint conditional probability distribution,
denoted by 
\[
p(a_{1},a_{2}\dots a_{n}|x_{1},x_{2}\dots x_{n}).
\]
Further, if $p(a_{1},a_{2}\dots a_{n}|x_{1},x_{2}\dots x_{n})=\tr\left[M_{a_{1}|x_{1}}^{1}\otimes M_{a_{2}|x_{2}}^{2}\dots\otimes M_{a_{n}|x_{n}}^{n}\rho\right]$
then we call the set of boxes, \emph{quantum boxes}, where $\{M_{a'|x'}^{i}\}_{a'\in\mathcal{A}_{i}}$constitute
a POVM for a fixed $i$ and $x'$, $\rho$ is a density matrix and
their dimensions are mutually consistent.
\end{defn}

Henceforth, we restrict ourselves to quantum boxes. 
\begin{defn}[Protocol in the box formalism]
 \label{def:BoxProtocol}A generic two-party protocol in the box
formalism has the following form:
\begin{enumerate}
\item Inputs:
\begin{enumerate}
\item Alice is given boxes $\Box_{1}^{A},\Box_{2}^{A}\dots\Box_{p}^{A}$
and Bob is given boxes $\Box_{1}^{B},\Box_{2}^{B},\dots\Box_{q}^{B}$. 
\item Alice is given a random string $r^{A}$ and Bob is given a random
string $r^{B}$ (of arbitrary but finite length).
\end{enumerate}
\item Structure: At each round of the protocol, the following is allowed.
\begin{enumerate}
\item Alice and Bob can locally perform arbitrary but finite time computations
on a Turing Machine. 
\item They can exchange classical strings/messages and boxes.
\end{enumerate}
\end{enumerate}
\branchcolor{black}{A protocol in the box formalism is readily expressed as a protocol
which uses a (trusted) classical channel (i.e. they trust their classical
devices to reliably send/receive messages), untrusted quantum devices
and an untrusted quantum channel (i.e. a channel that can carry quantum
states but may be controlled by the adversary).}
\end{defn}

\begin{assumption}[Setup of Device Independent Two-Party Protocols]
 Alice and Bob 
\begin{enumerate}
\item both have private sources of randomness,
\item can send and receive classical messages over a (trusted) classical
channel,
\item can prevent parts of their untrusted quantum devices from communicating
with each other, and
\item have access to an untrusted quantum channel.
\end{enumerate}
\end{assumption}

\branchcolor{black}{We restrict ourselves to a ``measure and exchange'' class of protocols---protocols
where Alice and Bob start with some pre-prepared states and subsequently,
only perform classical computation and quantum measurements locally
in conjunction with exchanging classical and quantum messages. More
precisely, we consider the following (likely restricted) class of
device independent protocols.}
\begin{defn}[Measure and Exchange (Device Independent Two-Party) Protocols]
\label{def:MEprotocol} A \emph{measure and exchange (device independent
two-party) protocol} has the following form:
\begin{enumerate}
\item Inputs:
\begin{enumerate}
\item Alice is given quantum registers $A_{1},A_{2},\dots A_{p}$ together
with POVMs\footnote{For concreteness, take the case of binary measurements. By $\{M_{a|x}^{A_{1}}\}_{a}$,
for instance, we mean $\{M_{0|x}^{A_{1}},M_{1|x}^{A_{1}}\}$ is a
POVM for $x\in\{0,1\}$. } 
\[
\{M_{a|x}^{A_{1}}\}_{a},\{M_{a|x}^{A_{2}}\}_{a},\dots\{M_{a|x}^{A_{p}}\}_{a}
\]
which act on them and Bob is, analogously, given quantum registers
$B_{1},B_{2},\dots B_{q}$ together with POVMs 
\[
\{M_{b|y}^{B_{1}}\}_{b},\{M_{b|y}^{B_{2}}\}_{b},\dots,\{M_{b|y}^{B_{q}}\}_{b}.
\]
Alice shields $A_{1},A_{2},\dots A_{p}$ (and the POVMs) from each
other and from Bob's lab. Bob similarly shields $B_{1},B_{2}\dots B_{q}$
(and the POVMs) from each other and from Alice's lab.
\item Alice is given a random string $r^{A}$ and Bob is given a random
string $r^{B}$ (of arbitrary but finite length).
\end{enumerate}
\item Structure: At each round of the protocol, the following is allowed.
\begin{enumerate}
\item Alice and Bob can locally perform arbitrary but finite time computations
on a Turing Machine.
\item They can exchange classical strings/messages.
\item Alice (for instance) can 
\begin{enumerate}
\item send a register $A_{l}$ and the encoding of her POVMs $\{M_{i}^{A_{l}}\}_{i}$
to Bob, or
\item receive a register $B_{m}$ and the encoding of the POVMs $\{M_{i}^{B_{m}}\}_{i}$. 
\end{enumerate}
Analogously for Bob. 
\end{enumerate}
\end{enumerate}
\end{defn}

It is clear that a protocol in the box formalism (\Defref{BoxProtocol})
which uses only quantum boxes (\Defref{box}) can be implemented as
a measure and exchange protocol (\Defref{MEprotocol}).
  
\end{document}